\documentclass[11pt]{article}
\usepackage[letterpaper]{geometry}
\usepackage{amsmath,amsthm,amsfonts,amssymb}
\usepackage{enumerate,color,xcolor}
\usepackage{graphicx}
\usepackage{subfigure}
\usepackage{url}
\usepackage{comment}
\numberwithin{equation}{section}
\theoremstyle{plain}

\newtheorem{theorem}{Theorem}

\newtheorem{corollary}[theorem]{Corollary}
\newtheorem{definition}[theorem]{Definition}

\newtheorem{lemma}[theorem]{Lemma}

\newtheorem{proposition}[theorem]{Proposition}

\newtheorem{remark}[theorem]{Remark}

\begin{document}

\begin{center}
  \Large \bf Asymptotics of the time-discretized log-normal SABR model:
The implied volatility surface
\end{center}

\author{}
\begin{center}
Dan Pirjol\,
\footnote{School of Business, Stevens Institute of Technology, Hoboken, NJ 07030, United States of America; dpirjol@gmail.com},
Lingjiong Zhu\,
\footnote{Department of Mathematics, Florida State University, 1017 Academic Way, Tallahassee, FL-32306, United States of America; zhu@math.fsu.edu}
\end{center}

\begin{center}
 \today
\end{center}

\begin{abstract}
We propose a novel time discretization for the log-normal SABR model
which is a popular stochastic volatility model that is 
widely used in financial practice. Our time discretization is a variant of the 
Euler-Maruyama scheme. We study its asymptotic
properties in the limit of a large number of time steps under a certain asymptotic regime
which includes the case of finite maturity, small vol-of-vol and large initial volatility
with fixed product of vol-of-vol and initial volatility.
We derive an almost sure limit and a large deviations result for the log-asset 
price in the limit of large number of time steps. We derive an exact
representation of the implied volatility surface for
\textit{arbitrary maturity and strike} in this regime.
Using this representation we obtain analytical expansions 
of the implied volatility for small maturity and extreme strikes,
which reproduce at leading order known asymptotic results for
the continuous time model.
\end{abstract}

\section{Introduction}
\label{sec:1}

The method of the asymptotic expansions has been used in the literature
to study the properties of stochastic volatility models under a wide
variety of strike-maturity regimes. The short maturity limit at fixed strike
for the implied volatility was first derived for the SABR model
in the celebrated work of \cite{SABR}. This model is defined by the 
two-dimensional stochastic differential equation:
\begin{eqnarray}\label{1}
&& dS(t) = \sigma(t) S^\gamma(t) dW(t)\,, \\
&& d\sigma(t) = \omega \sigma(t) dZ(t)\,, \nonumber
\end{eqnarray}
where $W(t)$ and $Z(t)$ are two standard Brownian motions, with $\mathbb{E}[dW(t)dZ(t)]=
\varrho dt$. $\omega>0$ is the volatility of volatility 
parameter and $\gamma \in [0,1]$ is an exponent which controls the backbone of 
the implied volatility \cite{SABR}. The correlation $\varrho$ can take any
values in $\varrho \in [-1,+1]$, although the most important cases for 
applications have $\varrho \leq 0$.

The short maturity asymptotic expansion of \cite{SABR} was extended in
\cite{BBF,HLbook} to a wider class of stochastic volatility models. 
Similar short maturity asymptotic results were obtained for other stochastic 
volatility models popular in financial practice such as the Heston model 
\cite{FJ0,FJ2} and in more general forms \cite{HLbook,Deuschel}.
Asymptotics in the large maturity regime were obtained, both at
fixed strike \cite{Lewis,Lewis2002,Rogers,Forde1} and in the 
joint strike-maturity regime
\cite{FJ1}, in a wide variety of stochastic volatility models. 
The option price asymptotics can be translated into implied volatility 
asymptotics using the transfer results of Gao and Lee \cite{GaoLee}.

The properties of the SABR model are well understood in continuous time.
The martingale properties of the model were studied in \cite{Jourdain,LM}.
Short maturity asymptotics for the asset price distribution, option prices
and implied volatility were first derived at leading order 
by Hagan et al \cite{SABR}.  The expansion was extended to higher order in
\cite{HLbook,Paulot,Lewis2017}. 
The asymptotics was studied using operator expansion methods in \cite{CCMN}.
For a survey of existing results see \cite{SABRbook}. A mean-reverting version 
of the model called $\lambda$-SABR 
was introduced in \cite{HLbook}; the asymptotics of options prices was
recently studied in \cite{GHN}.
The simulation and pricing under the SABR model have been studied
extensively. For $0 < \gamma < 1$ and zero correlation $\varrho=0$, 
an exact representation for the conditional distribution of $S_T$ 
for given $(\sigma_T , \int_0^T \sigma(t)^2 dt)$ was given by Islah \cite{Islah}.
This was further simplified by \cite{AntonovWings} who derived a 
one-dimensional integral representation for the option prices. 
Cai, Song and Chen \cite{CaiSongChen} gave an exact 
simulation method for the SABR model for $\gamma=1$ and for 
$\gamma < 1, \varrho=0$, using an inversion of the
Laplace transform for $1/\int_0^T \sigma(t)^2 dt$.

The log-normal ($\gamma=1$) SABR model is an important limiting case. 
This can be regarded as a particular case of Hull-White stochastic
volatility model \cite{HW1987}
\begin{equation}
dS(t) = \sqrt{V(t)} S(t) dW(t)\,,\quad 
dV(t) = \xi V(t) dZ(t) + \eta V(t) dt\,,\quad \xi>0, \eta < \frac14\xi^2\,.
\end{equation}
The process for the instantaneous variance 
$V(t) = \sigma^2(t)$ is equivalent to 
$d\sigma(t) = \frac12 \xi \sigma(t) dZ(t) + (\frac12 \eta - \frac18 \xi^2) 
\sigma(t) dt$,
which reduces to the volatility process in the SABR model when 
$\eta=\frac14\xi^2$.
Option price and implied volatility asymptotics in the Hull-White model
at large strike were studied by Gulisashvili and Stein \cite{GS1,GS2,Gulbook}.
The large maturity asymptotics in the SABR model were studied by Forde and 
Pogudin \cite{Forde1} and by Lewis in the Hull-White model \cite{Lewis2002}.
The log-normal SABR model can be mapped to the Brownian motion on the
Poincar\'e space $\mathbb{H}^3$ \cite{HLbook,Lewis2017}. 
The model can be simulated exactly by conditional Monte Carlo methods, 
exploiting the fact that $\log S_T$
is normally distributed, conditional on a realization of 
$(\sigma_T , \int_0^T \sigma(t)^2 dt)$.
While an exact result for this distribution is known from Yor \cite{Yor},
the numerical evaluation of the result is challenging, requiring very 
high precision in intermediate steps \cite{CaiSongChen}.
The paper \cite{CaiSongChen} presents an alternative approach 
involving the inversion of a Laplace transform. 
An integral representation of the pricing
kernel has been presented in \cite{Lewis2017}.

To the best of our knowledge, all the asymptotic results in the literature
are obtained in the continuous time context. In many practical applications
of these models, they are simulated in discrete time, by application of
time discretization schemes such as the Euler-Maruyama scheme.
We study in this paper the asymptotics of the model (\ref{1}) with $\gamma=1$
discretized
in time under an application of the Euler-Maruyama scheme to $\log S(t)$,
and appropriate model parameters rescaling with $n$, the number of time steps.

Consider a grid of time points $\{ t_i \}_{i=1}^n$ with 
uniform time step size $\tau$. Application of the Euler-Maruyama discretization 
of (\ref{1}) to $\log S(t)$ gives the stochastic recursion
\begin{eqnarray}\label{SLogEuler}
&& S_{i+1} = S_i e^{\sigma_i \Delta W_i - \frac12 \sigma_i^2 \tau}\,, \\
&& \sigma_i = \sigma_0 e^{\omega Z_i - \frac12 \omega^2 t_i} \,, \nonumber
\end{eqnarray}
where $\Delta W_{i}:=W(t_{i+1})-W(t_{i})$ and $Z_{i}:=Z(t_{i})$. 
This was called in \cite{PZStochVol} the Log-Euler-log-Euler scheme. 
Its asymptotic properties were studied in \cite{PZStochVol} in the limit 
$n \to \infty$ at fixed $\beta = \frac12\omega^2 \tau n^2, 
\rho = \sigma_0 \sqrt{\tau}$, in the uncorrelated limit $\varrho=0$. 
The main results obtained were an almost sure limit for the asset price
$\lim_{n\to \infty} \frac{1}{n} \log S_n = - \frac12 \rho^2$ a.s. (Proposition
19 \cite{PZStochVol}), a fluctuations result $\lim_{n\to\infty} 
\frac{\log S_n + \frac12\rho^2 n}{\sqrt{n}} = N(0,\rho^2+\frac23
\rho^4 \beta)$ in distribution (Proposition 20 \cite{PZStochVol}), and a closed form
result for the Lyapunov exponents of the asset price moments
$\lambda(\rho,\beta;q) = \lim_{n\to \infty} \frac{1}{n} \log\mathbb{E}[(S_n)^q]$.
It was pointed out that the scheme (\ref{SLogEuler}) differs from the 
continuous time model in one notable respect: the asset price is a martingale 
for any correlation $\varrho\in  [-1,1]$. Recall that in continuous time this property holds only 
for non-positive correlation \cite{BCM,Jourdain,LM}.

In this paper we introduce an alternative time discretization which is more
tractable for the non-zero correlation case $\varrho \neq 0$.
In addition we show that this scheme reproduces the martingale properties
of the continuous time model as $n\to \infty$: the asset price is
a martingale provided that $\varrho \leq 0$, see 
Proposition 8. Numerical study shows that it produces a martingale defect for 
$\varrho >0$. 
This scheme is defined in (\ref{scheme2}) and reduces to the 
Log-Euler-log-Euler scheme (\ref{SLogEuler}) in the uncorrelated limit 
$\varrho=0$.

We study in this paper the asymptotics of the new scheme in the limit of a large
number of time steps $n\to \infty$ at fixed 
$\beta = \frac12\omega^2 \tau n^2,\rho = \sigma_0 \sqrt{\tau}$. 
We derive a large deviations property for the log-price of the 
asset $\mathbb{P} (\frac{1}{n}\log S_n \in \cdot)$ in this limit. 
The large deviations result is translated into option prices 
and implied volatility asymptotics. 
The rate function turns out to be independent of the time step $\tau$.
The limit considered includes the regime of 
finite maturity $T = n\tau = O(1)$, small vol of vol  
$\omega^2 T \to 0$ and large initial vol $\sigma_0^2 T \to \infty$ at 
fixed $(\omega^2 T)(\sigma_0^2 T)$. 
We obtain the volatility surface of the model for arbitrary 
maturity $T$ and strike $K$ under this regime in Theorem~\ref{thm:main}:
\begin{equation}\label{1.4}
\sigma_{\rm BS}(x,T)= 
\sigma_{0}\Sigma_{\rm BS}\left( \frac{x}{\sigma_0^2 T};
2(\sigma_0^2 T)(\omega^2 T)\right)\,,\quad x = \log\frac{K}{S_0}\,,
\end{equation}
where the equality above means the ratio LHS/RHS goes to one in the limit 
considered. 
The function $\Sigma_{\rm BS}(y;a)$ is given in explicit form in 
Theorem~\ref{thm:main}. To the best of our knowledge this is the first 
stochastic volatility model for which the entire volatility surface can be 
approximated in closed form in a certain limit of the model parameters.

The limiting result for the volatility surface given by Theorem~\ref{thm:main}
is compared with known asymptotic expansions of the implied volatility in the
log-normal SABR model in the short-maturity \cite{SABR,HLbook,PaulotSABR}, 
large-maturity \cite{Lewis2002,Lewis,Forde1} and
extreme strikes \cite{GSHW} regimes. The asymptotic result of (\ref{1.4}) 
reproduces all these expansions, after taking
the small vol of vol and large initial volatility limit. 
A summary of the known asymptotic results in the continous
time model and their counterparts in the discrete-time model is given in 
Table~\ref{tab:summary}.

\begin{table}
\caption{\label{tab:summary} 
Comparison of the known asymptotic results for the continuous-time model
with the discrete-time asymptotics obtained in this paper, showing the relevant
references and sections where they are discussed. In all cases we reproduce
the known asymptotic results for the continuous-time model.}
\begin{center}
\begin{tabular}{|c|c|c|}
\hline
asymptotics  & continuous time & discrete time  \\
\hline \hline
$T\to 0$ & $O(T^0)$: Hagan et al \cite{SABR}
         & Sec.~6.1 \\
         & $O(T), O(T^2)$: Paulot \cite{Paulot}
         &  \\
\hline
$T\to 0, x\to 0$ & double expansion $O(T^k x^n)$ Lewis \cite{Lewis,Lewis2017}
         & Sec.~6.2 $n=0,k\leq 4$ \\
\hline
extreme strikes & $|x|\to \infty$: Gulisashvili, Stein ($\varrho=0$) \cite{GSHW}
         & Sec.~6.3 \\
\hline 
\end{tabular}
\end{center}
\end{table}

We comment also on the relation of our results to those obtained by 
Forde \cite{FordeLargeTime} and Forde and Kumar \cite{FordeKumar} for the large 
maturity asymptotics in stochastic volatility models. In the simplest setting 
these papers study the class of models defined by
\begin{eqnarray}
&& \frac{dS(t)}{S(t)} = \sigma(Y(t)) dW^1(t)\,, \\
&& dY(t) = - \alpha Y(t) dt + dW^2(t)\,,
\end{eqnarray}
where $W^1(t),W^2(t)$ are standard Brownian motions that may be correlated. The method is based on proving
large deviations for $\mathbb{P}(A(t) \in \cdot)$ as $t\rightarrow\infty$ 
for the time average of the integrated variance
\begin{equation}
A(t) = \frac{1}{t} \int_0^t \sigma^2 (Y(s))ds\,.
\end{equation}
This is done using the Donsker-Varadhan type large deviations \cite{DonskerVaradhan} for the
occupation time of the $Y(t)$ process, and then applying the contraction
principle. The rate function for the large deviation principle of $\mathbb{P}(A(t) \in \cdot)$ 
is denoted in \cite{FordeKumar} as $I_f(a)$. 
Assuming for simplicity zero correlation as in \cite{FordeLargeTime},
the log-asset price
$X(t) = \log S(t)$ is related to $A(t)$ as
\begin{equation}
\frac{1}{t} X(t) = - \frac12 A(t) + \frac{Z}{\sqrt{t}} \sqrt{A(t)}\,,
\end{equation}
with $Z\sim N(0,1)$ following a standard normal distribution. 
This is similar to our Eq.~(\ref{scheme2}) in the zero correlation limit $\varrho=0$.
An application
of the contraction principle gives that $\mathbb{P}(\frac{1}{t} X(t) \in
\cdot)$ satisfies a large deviation principle with rate function
\begin{equation}
I(x) = \inf_{a>0} \left( I_f(a) + \frac{1}{2a} \left(x + \frac12 a\right)^2\right)\,.
\end{equation}

Our approach differs from that of \cite{FordeLargeTime,FordeKumar} in two 
respects. First, the analog of the rate function of the integrated variance 
$I_f(a)$ is obtained exactly as in \cite{PZDiscreteAsian}, 
without requiring the Donsker-Varadhan large deviations results \cite{DonskerVaradhan}. 
This allows a more explicit treatment, and for our particular model, 
we do not rely on numerical methods. Second, we do not 
require the large maturity limit, and our asymptotic regime can include 
arbitrary maturity.

The paper is organized as follows. In Section~\ref{sec:2} we introduce the
new time discretization scheme for the asset price. In Section~\ref{sec:3}
we study the $n\to \infty$ asymptotics for
the volatility process. These results are used in Section~\ref{sec:4}
to study the asymptotics of the asset price process. We derive an almost sure 
limit and large deviations for $\mathbb{P}(\frac{1}{n}\log S_n \in \cdot)$.
These results are used in Section~\ref{sec:5} to obtain option price asymptotics
and the implied volatility asymptotics in the $n\to \infty$ limit.
Section~\ref{sec:6} studies in detail the 
implications of these results in various regimes of small maturity and extreme 
strikes. 
Finally Section~\ref{sec:7} compares the asymptotic result against numerical 
benchmarks.  
An Appendix derives the rate functions for the uncorrelated model in explicit
form, and obtains asymptotics in various regimes of small/large arguments.


\section{Setup of the time discretized model}
\label{sec:2}

\begin{definition}[Modified Log-Euler, Log-Euler scheme]\label{def:1}
Assume the timeline $\{ t_i \}_{i=0}^n$ with uniform time step $t_{i+1}-t_{i}=\tau$, and
denote for simplicity $S(t_i) = S_i, \sigma(t_i)=\sigma_i$.
The discretization scheme of the $\gamma=1$ SABR model \eqref{1} is defined 
recursively by
\begin{equation}\label{scheme2iter}
\log S_{i+1} = \log S_{i} + \varrho_\perp \sigma_i \Delta W_i^\perp - \frac12 \sigma_i^2\tau 
+ \varrho \frac{1}{\omega} (\sigma_i - \sigma_0),
\end{equation}
with $\sigma_i = \sigma_0 e^{\omega Z_i - \frac12 \omega^2 t_i}$ and 
$\Delta W_i^\perp = W^\perp(t_{i+1}) - W^\perp(t_i)$ 
where $W^\perp(t)$ is a Brownian motion independent of $Z(t)$, given by
$W(t) = \varrho Z(t) + \sqrt{1-\varrho^2} W^\perp(t)$.
\end{definition}

It follows from Definition \ref{def:1} that 
the recursion can be written in closed form as
\begin{align}\label{scheme2}
\log S_n &= \log S_0 + \varrho_\perp \sum_{i=0}^{n-1} \sigma_i \Delta W_i^\perp - \frac12 \sum_{i=0}^{n-1}
\sigma_i^2 \tau + \varrho\frac{1}{\omega} (\sigma_n - \sigma_0) \\
&= \log S_0 + \varrho_\perp \sqrt{V_n} Z - \frac12 V_n +  \varrho\frac{1}{\omega} (\sigma_n - \sigma_0)\,, \nonumber
\end{align}
where $V_n =\sum_{i=0}^{n-1}
\sigma_i^2 \tau$ and $Z\sim N(0,1)$ follows a standard normal distribution, 
independent of $\sigma_i$.
The construction of this scheme is motivated by writing the SDE of the continuous time model in
terms of  $X(t) = \log S(t)$ as
\begin{eqnarray}\label{SDE3}
dX(t) &=& \sigma(t) dW(t) - \frac12 \sigma(t)^2 dt 
= \varrho \sigma(t) dZ(t) + \varrho_\perp \sigma(t) dW^\perp(t) - 
\frac12 \sigma(t)^2 dt  \\
&=& \varrho \frac{1}{\omega} d\sigma(t) + \varrho_\perp \sigma(t) dW^\perp(t) - 
\frac12 \sigma(t)^2 dt \,, \nonumber
\end{eqnarray}
where we denoted $\varrho_\perp:= \sqrt{1-\varrho^2}$. 
In this form we see that $X(t)$ can be decomposed as the sum of two independent processes
$X(t) = X^\perp(t) + X^\parallel(t)$ with 
\begin{equation}
dX^\perp(t) = \varrho_\perp \sigma(t) dW(t)^\perp - 
\frac12 \sigma(t)^2 dt\,,\quad X^\perp_0 = \log S_0\,,
\end{equation}
and $X^\parallel(t) = \varrho \frac{1}{\omega} (\sigma(t) - \sigma_0)$. 
The discretization scheme (\ref{scheme2iter}) is obtained by applying Euler discretization to
$X^\perp(t)$ and keeping $X^\parallel(t)$ in closed form.

We will consider the properties of the time-discretized model (\ref{scheme2})
in  the $n \to \infty$ limit at fixed 
\begin{equation}\label{betarho}
\beta:= \frac{1}{2}\omega^{2}\tau n^{2},
\qquad
\rho:=\sigma_{0}\sqrt{\tau}.
\end{equation} 
This limit covers the following
asymptotic regimes of practical interest:
\begin{itemize}

\item Finite maturity, low vol of vol and large initial volatility.
This corresponds to $t_n = n\tau = O(1), \omega = O(n^{-\frac12}), 
\sigma_0 = O(n^\frac12)$. Taking 
$\omega=\tilde \omega n^{-\frac12},
\sigma_0 = \tilde \sigma_0 \sqrt{n}$ and fixed $T=n\tau$, we have
$\beta=\frac12 \tilde\omega^2 T, \rho = \tilde \sigma_0 \sqrt{T}$, 
independent of $n$.

\item Small maturity regime and large initial volatility. This corresponds 
to $\sigma_0 = O(n), \tau =O(n^{-2})$ and $\omega=O(1)$. This gives
$t_n =n\tau=O(n^{-1})$.

\item Large maturity regime and low vol of vol. This corresponds to the
situation when $\sigma_0,\tau$ are fixed, vol of vol $\omega=O(1/n)$,
and the maturity $t_n = n\tau =O(n)$.

\end{itemize}

Under the first two regimes the time step goes to 
zero $\tau\to 0$ in the $n\to \infty$ limit, and thus they are appropriate
for studying the continuous time limit of this model. The first regime is 
studied in more detail in Sections~\ref{sec:5} and \ref{sec:6}. 
The asymptotic predictions under this regime are compared with known asymptotic
results in the continuos time limit. In all cases where an analytical result is 
known for the continuous time case, we recover it in the $n\to \infty$ limit 
of the discrete time model.

\section{Large deviations for the volatility process}
\label{sec:3}

The time discretized asset price under the scheme (\ref{scheme2}) is the sum 
of two terms
\begin{equation}\label{sum12}
\frac{1}{n} \log S_n = \frac{1}{n}\log S_n^\perp + 
\varrho \frac{1}{n\omega} (\sigma_n - \sigma_0) \,.
\end{equation}
Conditional on a path of the volatility process 
$\{\sigma_k\}_{k=1}^n$, the
asset price $S_n$ is log-normally distributed. The distribution of $S_n$
is obtained by folding the log-normal distribution with the distribution
of the volatility process. It will be seen that the large deviations
properties of $\frac{1}{n} \log S_n$ are related to the large deviations
for the joint volatility process $\{ V_i, \sigma_i\}_{i=1}^n$.


We start by reviewing the known results for the $n\to \infty$ asymptotics of 
the volatility process. It was shown in \cite{PZDiscreteAsian} that for 
$n\to \infty$ at fixed $\beta,\rho$, one has the almost sure limit
\begin{equation}\label{Vnas}
\lim_{n\to \infty} \frac{1}{n} V_n = \rho^2 \mbox{ a.s.}
\end{equation}
A similar argument gives the almost sure limit for the volatility process
\begin{equation}\label{signas}
\lim_{n\to \infty} \frac{\sigma_n}{n\omega} =  
v_0 := \frac{\rho}{\sqrt{2\beta}} \mbox{ a.s.}
\end{equation}

The large deviations for $\mathbb{P}(\frac{1}{n} V_n\in
\cdot)$ were obtained in Proposition 6 of \cite{PZDiscreteAsian}.
This result is summarized for convenience in the Appendix~\ref{sec:app}.
The first term in (\ref{sum12}) depends on $\{ \sigma_0,\sigma_1,\cdots , \sigma_{n-1}\}$
through $V_n$, while the second term depends on $\sigma_n$. 
This introduces correlation among the two terms. This requires that we study
the large deviations properties of the joint process $(V_n,\sigma_n)$ as 
$n\to \infty$.


We start by introducing some background concepts about large deviations theory.
We refer to Dembo and Zeitouni \cite{Dembo} for a more comprehensive exposition of
large deviations and its applications. A sequence $(P_{n})_{n\in\mathbb{N}}$ 
of probability measures on a topological space $X$ 
is said to satisfy a \textit{large deviation principle} (LDP) with the 
\textit{rate function} $I:X\rightarrow\mathbb{R}\cup\{\infty\}$ if $I$ is 
non-negative, lower semicontinuous and for any measurable set $A$, we have
\begin{equation}
-\inf_{x\in A^{o}}I(x)\leq\liminf_{n\rightarrow\infty}\frac{1}{n}\log P_{n}(A)
\leq\limsup_{n\rightarrow\infty}\frac{1}{n}\log P_{n}(A)\leq-\inf_{x\in\overline{A}}I(x).
\end{equation}
Here, $A^{o}$ is the interior of $A$ and $\overline{A}$ is its closure. 
The rate function $I$ is said to be \textit{good} if for any $m$, the level set
$\{x:I(x)\leq m\}$ is compact.

The \textit{contraction principle} (e.g. Theorem 4.2.1. \cite{Dembo})
plays a key role in our proofs. For the convenience
of the readers, we state the contraction principle as follows.
If $P_{n}$ satisfies a large deviation principle on $X$ with rate 
function $I(x)$ and $F:X\rightarrow Y$ is a continuous map,
then the probability measures $Q_{n}:=P_{n}F^{-1}$ satisfies
a large deviation principle on $Y$ with rate function
$J(y)=\inf_{x: F(x)=y}I(x)$.


We start by noting that the last term in (\ref{sum12}) satisfies a LDP. 

\begin{lemma}\label{prop:2}
In the limit $n\to \infty$,
$\mathbb{P}(\frac{1}{n\omega} \sigma_n \in \cdot)$ satisfies a LDP
with rate function
\begin{equation}\label{I2}
\mathcal{I}_2(x) = \frac{1}{4\beta} \log^2 \left(\frac{x}{v_0}\right)\,,
\end{equation}
for $x>0$, with $v_0 = \frac{\rho}{\sqrt{2\beta}}$, and $\mathcal{I}_2(x)=\infty$
otherwise.
\end{lemma}

\begin{proof}
Start with $\frac{1}{n\omega} \sigma_n =\frac{\sigma_0}{n\omega} e^{\omega Z_n-\frac12
\omega^2 t_n}$.
We have $\omega Z_n = \omega \sqrt{\tau n} Z=\frac{\sqrt{2\beta}}{\sqrt{n}}Z$ with $Z\sim N(0,1)$,
and $\frac{1}{2}\omega^{2}t_{n}=\frac{1}{2}\omega^{2}\tau n=\frac{\beta}{n}$.
Since $Z\sim N(0,1)$, it follows that for any $\theta\in\mathbb{R}$,
\begin{equation}
\lim_{n\rightarrow\infty}\frac{1}{n}\log\mathbb{E}\left[e^{\theta n(\omega Z_n-\frac12
\omega^2 t_n)}\right]
=\frac{1}{2}(2\beta)\theta^{2},
\end{equation}
which by G\"{a}rtner-Ellis theorem \cite{Dembo} implies that
$\mathbb{P}(\omega Z_n-\frac12
\omega^2 t_n\in\cdot)$ satisfies a large deviation principle
with rate function $\sup_{\theta\in\mathbb{R}}\{\theta x-\frac{1}{2}(2\beta)\theta^{2}\}
=\frac{1}{4\beta}x^{2}$. 
Note that 
$\frac{\sigma_{0}}{n\omega} =\frac{\rho\sqrt{\tau}n}{n\sqrt{\tau}\sqrt{2\beta}}=\frac{\rho}{\sqrt{2\beta}}$,
and thus $\frac{1}{n\omega}\sigma_{n}=\frac{\rho}{\sqrt{2\beta}}e^{\omega Z_{n}-\frac{1}{2}\omega^{2}t_{n}}$.
By the contraction principle \cite{Dembo},
$\mathbb{P}(\frac{1}{n\omega} \sigma_n \in \cdot)$ satisfies a LDP
with rate function $\frac{1}{4\beta}\log^{2}(\frac{x}{\rho/\sqrt{2\beta}})$. 
\end{proof}


Next we prove a LDP for the joint 
distribution of $(V_n,\sigma_{n})$,
for $\mathbb{P}((\frac{1}{n} V_n, \frac{1}{n\omega} \sigma_{n}) \in \cdot )$.
The proof proceeds in close analogy with Proposition 6 in \cite{PZDiscreteAsian}.

\begin{proposition}\label{prop:I}
Consider the $n\to \infty$ limit at fixed $\beta = \frac12 \omega^2 \tau n^2$
and $\rho = \sigma_0 \sqrt{\tau}$.
Define $V_n=\sum_{k=0}^{n-1} \sigma_k^2 \tau$. Then
$\mathbb{P}((\frac{1}{n} V_n, \frac{1}{n\omega} \sigma_{n}) \in \cdot )$
satisfies a LDP with rate function
\begin{equation}\label{I2rf}
\mathcal{I}(x,y) = \inf_{g\in \mathcal{AC}_{x,y}[0,1]}
\frac12 \int_0^1 \left( g'(z) \right)^2 dz\,,
\end{equation}
for $x,y \geq 0$, and $\mathcal{I}(x,y)=\infty$ otherwise.
We denoted
\begin{equation}
\mathcal{AC}_{x,y}[0,1] = \left\{ g \bigg| g\in AC[0,1], g(0)=0,
\int_0^1 e^{2\sqrt{2\beta} g(z)} dz=\frac{x}{\rho^2},
e^{\sqrt{2\beta} g(1)}=y \frac{\sqrt{2\beta}}{\rho} \right\}\,,
\end{equation}
where $AC[0,1]$ denotes the space of absolutely continuous functions from $[0,1]$ to $\mathbb{R}$.
\end{proposition}

\begin{proof}
Write $\sigma_k$ as
\begin{equation}
\sigma_k = \sigma_0 e^{\omega Z_k - \frac12 \omega^2 t_k}  =
\sigma_0 e^{\omega \sqrt{\tau} \sum_{j=1}^{k-1} U_j - \frac12 \omega^2 k \tau}\,,
\end{equation}
where $U_j := \frac{1}{\sqrt{\tau}} (Z_j - Z_{j-1})$ are i.i.d. $N(0,1)$ random 
variables. 
Using $\omega^2 \tau = \frac{2\beta}{n^2}$, we get 
\begin{equation}
\sigma_k  =
\sigma_0 e^{\frac{\sqrt{2\beta}}{n} \sum_{j=1}^{k-1} U_j - \frac{\beta}{n^2}k}\,.
\end{equation}
Note that $e^{-\frac{\beta}{n}} \leq e^{- \frac{\beta}{n^2}k} \leq 1$
uniformly in $0 \leq k \leq n-1$. Thus we can neglect the term $-\frac{\beta}{n^2}k$
in the exponent in the $n\to \infty$ limit.

By Mogulskii theorem \cite{Dembo}, $\mathbb{P}( \frac{1}{n} \sum_{j=1}^{[\cdot n]} U_j 
\in \cdot ) $ satisfies a LDP on $L^\infty[0,1]$ with good rate function
\begin{equation}
I(g) = \int_0^1 \Lambda(g'(x)) dx\,,
\end{equation}
if $g\in AC[0,1]$ and $g(0)=0$ and $I(g)=+\infty$ otherwise,
where 
\begin{equation}
\Lambda(x) := \sup_{\theta \in \mathbb{R}} \left\{ \theta x - \log
\mathbb{E}[e^{\theta U_1}] \right\} = \frac12 x^2.
\end{equation}

Let $g(x) := \frac{1}{n} \sum_{j=1}^{[xn]} U_j$. Then
\begin{equation}
\int_0^1 e^{2\sqrt{2\beta} g(x) } dx = \sum_{k=0}^{n-1}
\int_{k/n}^{(k+1)/n} e^{2\sqrt{2\beta} g(x) } dx =
\frac{1}{n} \sum_{k=0}^{n-1} e^{\frac{2\sqrt{2\beta}}{n} \sum_{j=1}^k U_j }\,.
\end{equation}

We know that $\mathbb{P}(\frac{1}{n\omega} \sigma_n \in \cdot )$
satisfies a LDP as shown above in Proposition \ref{prop:2} with rate function
$I_2(x)$ given in Equation (\ref{I2}). 
This random  variable is expressed as
\begin{equation}
\frac{1}{n\omega} \sigma_n = \frac{\rho}{\sqrt{2\beta}} 
e^{\frac{\sqrt{2\beta}}{n} \sum_{j=1}^{n-1} U_j - \frac{\beta}{n} },
\end{equation}
which is exponentially equivalent to $\frac{\rho}{\sqrt{2\beta}} 
e^{\sqrt{2\beta} g(1)}$ as $n\rightarrow\infty$.
We can apply the contraction principle \cite{Dembo} to conclude that
$\mathbb{P}( (\frac{1}{n} \sum_{k=0}^{n-1} \sigma_k^2 \tau , 
\frac{1}{n\omega} \sigma_n ) \in \cdot ) $ satisfies a LDP, with 
rate function $\mathcal{I}(x,y)$ given by the constrained
variational problem (\ref{I2rf}).
\end{proof}

\subsection{Explicit solution for the rate function $\mathcal{I}(u,v)$}

We present in this Section the solution of the variational problem in Eq.~(\ref{I2rf}) for the rate 
function $\mathcal{I}(u,v)$, which is given by
\begin{equation}
\mathcal{I}(u,v) = \inf_g \frac12 \int_0^1 [g'(t)]^2 dt\,,
\end{equation}
where the infimum is taken over all functions $g\in AC[0,1]$ satisfying
$g(0)=0$ and
\begin{equation}
\int_0^1 e^{2b g(t)} dt = \frac{u}{\rho^2}\,,\quad
e^{b g(1)} = \frac{vb}{\rho}\,,\quad
b := \sqrt{2\beta}\,.
\end{equation}

\begin{proposition}
The rate function can be expressed as
\begin{equation}\label{Iuvscale}
\mathcal{I}(u,v) = 
\frac{1}{16\beta} I(\bar u, \bar v)\,,
\end{equation}
where 
$\bar u := \frac{u}{\rho^2}$,
$\bar v := \frac{v}{v_0}$, and $v_0 := \frac{\rho}{\sqrt{2\beta}}$.
We distinguish three cases:

(i) $\bar u/\bar v>1$. For this case the rate function is
\begin{equation}\label{rf1}
I(\bar u,\bar v) = \beta \left\{ \beta 
- \frac{4\gamma}{\gamma+1} + \frac{4\gamma}{e^\beta + \gamma} \right\}\,,
\end{equation}
where $\beta$ is the solution of the equation
\begin{equation}\label{eqbeta}
\frac{\sinh(\beta/2)}{\beta/2} = \frac{\bar u }{\bar v}\,,
\end{equation}
and $\gamma$ is given by $\gamma = e^{\beta/2} \frac{\bar v e^{\beta/2} - 1}{e^{\beta/2} - \bar v}$.

(ii) $\bar u/\bar v<1$. For this case the rate function is
\begin{equation}
I(\bar u,\bar v) =  4\lambda^2
\left\{ \frac{\bar u}{\cos^2\eta}-1 \right\}\,,
\end{equation}
where $\lambda$ is the solution of the equation
\begin{equation}\label{eqlam}
\frac{\sin \lambda}{\lambda} = \frac{\bar u }{\bar v}\,,
\end{equation}
and $\eta$ is given by
$\tan\eta = \frac{\bar v\cos\lambda - 1}{\bar v\sin\lambda}$.

(iii) $\bar u = \bar v$. For this case the rate function is
\begin{equation}\label{rf3}
I(\bar u,\bar u) = 4 \frac{(\bar u-1)^2}{\bar u}\,.
\end{equation}
\end{proposition}

At the point $\bar u = 1, \bar v=1$, the optimal 
path is constant $g(t)=0$ and the rate function vanishes, in agreement with
the almost sure limit
$\mathcal{I}(\rho^2,v_0) =  0$.

\subsection{Alternative expression and approximation}

The function $I(u,v)$ can be expressed in an alternative form as
\begin{equation}\label{FHW}
I(u,v) = 8 F\left(\frac{v}{u}\right) + 4 \frac{1+v^2}{u} - 4\pi^2\,,
\end{equation}
where the function $F(\rho)$ is defined as
\begin{equation}\label{Fsol}
F(\rho) = \begin{cases}
 \frac12 x_1^2- \rho \cosh x_1 + \frac{\pi^2}{2}, &  0 < \rho < 1, \\
-\frac12 y_1^2+ \rho \cos y_1 + \pi y_1, & \rho > 1.
\end{cases}
\end{equation}
In Eqn. \eqref{Fsol}, $x_1, y_1$ are the solution of the equations
\begin{equation}
\rho \frac{\sinh x_1}{x_1} = 1\,,
\qquad
y_1 + \rho\sin y_1 = \pi\,.
\end{equation}
This function appears in the small-$t$ expansion of the Hartman-Watson 
distribution, and its properties are studied in more detail in \cite{HWexp}.
The function $F(\rho)$ has a minimum at $\rho=\frac{\pi}{2}$, with
$F(\frac{\pi}{2})=\frac{3\pi^2}{8}$. 
We will require the expansion of $F(\rho)$ around $\rho=1$ \cite{HWexp}.
\begin{equation}\label{Fexp}
F(\rho) = \frac{\pi^2}{2}-1 - (\rho-1) + \frac32 (\rho-1)^2
- \frac65 (\rho-1)^3 + O((\rho-1)^4)\,.
\end{equation}
Using this expansion we can derive an approximation for $I(u,v)$ around
its minimum at $u=v=1$.

\begin{proposition}\label{prop:Iuvexp}
Denote $\epsilon = \log u$ and $\eta= \log v$. 
The expansion of the rate function $I(u,v)$ around $u=v=1$ up to and
including quartic terms in $\epsilon,\eta$ is
\begin{align}\label{Iquart}
I(u,v) &=12\epsilon^2 - 24 \epsilon\eta + 16 \eta^2 
- \frac{12}{5}  \epsilon^3 + \frac{36}{5} \epsilon^2 \eta -
\frac{56}{5} \epsilon \eta^2 + \frac{32}{5}\eta^3  \nonumber \\
& \qquad\qquad + \frac{109}{175}\epsilon^4 -
\frac{436}{175} \epsilon^3 \eta
+ \frac{1004}{175} \epsilon^2 \eta^2 - \frac{1136}{175} 
\epsilon \eta^3 + \frac{1552}{525}\eta^4 + \cdots\,, 
\end{align}
where the ellipses denote terms of the form $\epsilon^a \eta^b$ with $a+b\geq 5$.
\end{proposition}

\begin{proof}
Follows from (\ref{FHW}) after using the expansion (\ref{Fexp}).
\end{proof}

The quadratic approximation 
\begin{equation}\label{Iquad}
I_q(u,v) := 12\log^2 u -24 \log u \log v+ 16 \log^2 v 
\end{equation}
gives a good approximation for $(u,v)$ sufficiently close to $(1,1)$. 

\subsection{One-dimensional projections of $\mathcal{I}(u,v)$}

From the contraction principle, we have
\begin{eqnarray}\label{Iconstu}
&& J_2(v)=\inf_{u} \mathcal{I}(u,v) 
   = \frac{1}{4\beta} \log^2 \bar v\,, \\
&& J_1(u)=\inf_{v} \mathcal{I}(u,v)  
   = \frac{1}{8\beta} \mathcal{J}_{BS}(\bar u)\,, \label{Iconstv}
\end{eqnarray}
where $J_1(u)$ and $J_2(v)$ are the rate functions of
large deviations of 
$\mathbb{P}(\frac{1}{n} V_n \in \cdot)$ and $\mathbb{P}(\frac{1}{n\omega} \sigma_n \in \cdot)$ 
computed respectively above in Lemma~\ref{prop:2}. 
Expressed in terms of $I(\bar u,\bar v)$ these relations read
\begin{equation}
\inf_{\bar u} I(\bar u,\bar v) = 4 \log^2 \bar v\,,
\qquad
\inf_{\bar v} I(\bar u,\bar v) = 2 \mathcal{J}_{BS}(\bar u)\,.
\end{equation}
The variables $\bar u>0,\bar v >0$ take positive real values,
just as the original variables $u>0, v>0$. They are rescaled such that
$I(\bar u, \bar v)=0$.

The rate function $\mathcal{J}_{BS}(u)$ was given in Proposition~6 in 
\cite{PZDiscreteAsian}
and Corollary 13 of the same paper. The result is reproduced in Proposition
\ref{prop:JBS} for convenience.


\section{Asymptotics for the asset price process $S_n$}
\label{sec:4}

Using the results of the previous section, we study here the asymptotics of 
$\frac{1}{n}\log S_n$ for the time discretization scheme (\ref{scheme2})
in the $n\to \infty$ limit at fixed $\beta = \frac12 \omega^2 \tau n^2$ and 
$\rho = \sigma_0 \sqrt{\tau}$.

\subsection{Almost sure limit}

\begin{proposition}\label{prop:LLN}
We have the almost sure limit 
\begin{equation}\label{LLN}
\lim_{n\to \infty} \frac{1}{n} \log S_n = - \frac12  \rho^2 \mbox{ a.s.}
\end{equation}
\end{proposition}

\begin{proof}
From (\ref{scheme2}) we have
\begin{align}
\frac{1}{n} \log S_n &= \frac{1}{n} \log S_n^\perp 
+ \varrho \frac{1}{\omega n} (\sigma_{n}-\sigma_0) \\
&= \frac{1}{n} \log S_0 + \varrho_\perp \sqrt{\frac{V_n}{n}} \cdot 
\frac{Z}{\sqrt{n}} - \frac{1}{2n} V_n + \varrho \frac{1}{\omega n} 
(\sigma_{n} - \sigma_{0})\,.
\nonumber
\end{align}
Taking the $n\to \infty$ limit of this relation and using the almost sure limit 
$\lim_{n\to \infty} \frac{1}{n} V_n = \rho^2 $, see Proposition 1 in
\cite{PZDiscreteAsian}, we get
$\lim_{n\to \infty} \frac{1}{n} \log S_n^\perp = - \frac12  \rho^2$
a.s.

Using $n\omega = \sqrt{\frac{2\beta}{\tau}}$, we get 
$\lim_{n\to \infty} \frac{1}{\omega n} \sigma_n = 
\lim_{n\to \infty} \sqrt{\frac{\tau}{2\beta}} 
\sigma_0 e^{\omega Z_n - \frac12 \omega^2 t_n} = \frac{\rho}{\sqrt{2\beta}}$
a.s.
and
$\lim_{n\to \infty} \frac{1}{\omega n} \sigma_0 = \frac{\rho}{\sqrt{2\beta}}$
a.s.
The two terms in the difference have the same limit so the contribution of
the last term  in (\ref{sum2}) cancels. 
\end{proof}

\begin{remark}\label{rem:4}
The discretization scheme (\ref{scheme2}) and the almost sure limit of 
Proposition \ref{prop:LLN} are easily extended by allowing a drift 
for the volatility process 
\begin{equation}
d\sigma(t) = \alpha \sigma(t) dt + \omega \sigma(t) dZ(t)\,,
\end{equation}
which can be solved as
\begin{equation}
\sigma_i = \sigma_0 e^{\omega Z_i + (\alpha - \frac12\omega^2)t_i}\,.
\end{equation}
Taking the large $n$ limit at fixed $\alpha_\infty:=\alpha \tau n $,
the result of Proposition~\ref{prop:LLN}  can be modified to take into account 
the drift term, as
\begin{equation}
\lim_{n\to \infty} \frac{1}{n} \log S_n = - \frac12 \rho^2 
\frac{e^{2\alpha_\infty}-1}{2\alpha_\infty} \mbox{ a.s.}
\end{equation}
The driftless Hull-White model \cite{HW1987} is defined 
by $dS(t) = \sqrt{V(t)} S(t) dW(t), dV(t) = \xi V(t) dZ(t)$, 
and is equivalent to the
process $d\sigma(t) = \frac12 \xi \sigma(t) dZ(t) - \frac18 \xi^2 \sigma(t) dt$.

The appropriate asymptotic limit is $n\to \infty$ at fixed $\beta = \frac12
\omega^2 \tau n^2, \rho = \sigma_0\sqrt{\tau}$ with $\omega=\frac12\xi$. 
In this limit $\alpha_\infty=-\frac18 \xi^2\tau n = O(1/n)$ which vanishes
for $n\to \infty$. We conclude that our results apply also to the Hull-White 
model with the replacement $\omega \to \frac12\xi$.
\end{remark}


\subsection{Asymptotic martingale property}

As $n\rightarrow\infty$, the asset price $S_n$ under the scheme
(\ref{scheme2}) is asymptotically a martingale for non-positive correlation $\varrho \leq 0$
in the following sense.

\begin{proposition}
For non-positive correlation $\varrho \leq 0$,
we have as $n\to \infty$ at fixed $n\tau = T$
\begin{equation}\label{S0m}
\lim_{n\to \infty} \mathbb{E}[S_n] = S_0 \mathbb{E}\left[e^{-\frac12\varrho^2 V_T
+ \varrho \frac{1}{\omega} (\sigma_T - \sigma_0)}\right] = S_0 \,,
\end{equation}
with $V_T := \int_0^T \sigma^2(t)  dt$.
\end{proposition}

\begin{proof}
From (\ref{scheme2}) we have
\begin{equation}\label{sum2}
S_n = S_0\exp\left(  \varrho_\perp \sqrt{V_n}
Z - \frac{1}{2} V_n +  \frac{\varrho}{\omega } 
(\sigma_{n} - \sigma_{0})\right)\,.
\end{equation}
Conditioning on $\{\sigma_i\}_{i=0}^n$, the asset price is log-normally 
distributed. Taking the expectation over $Z$ gives
\begin{equation}
\mathbb{E}[S_n] = 
S_0 \mathbb{E}\left[e^{-\frac12 \varrho^2 V_n + \frac{\varrho}{\omega}
(\sigma_n - \sigma_0)}\right]\,.
\end{equation}

First we prove the convergence in $L_1$ norm
\begin{equation}
V_n = \tau \sum_{i=0}^{n-1} \sigma_i^2 \to V(T) = \int_0^T \sigma^2(t) dt \,,
\quad n \to \infty\,, \quad n\tau = T\,.
\end{equation}
This follows by adapting the proof of Theorem 13 in \cite{IME} which proves 
a similar convergence statement for the discrete sum of a geometric Brownian 
motion to an integral.
By the Markov inequality this implies $V_n \to V(T)$ in probability, and
thus 
\begin{equation}
e^{-\frac12 \varrho^2 V_n + \varrho \frac{1}{\omega}(\sigma_n-\sigma_0)}
\to e^{-\frac12 \varrho^2 V_T + \varrho \frac{1}{\omega}(\sigma_T-\sigma_0)}\,,
\end{equation}
in probability.
For $\varrho\leq 0$, the exponential is bounded from above by 
$e^{-\frac12 \varrho^2 V_n}$.
By the Lebesgue dominated convergence theorem we can exchange limit and 
expectation. Using the known result for the continuous time case 
\cite{BCM,Jourdain,LM}, we get (\ref{S0m}).
\end{proof}

Numerical testing shows that for positive correlation there is a martingale 
defect $\mathbb{E}[S_n] < S_0$, which agrees numerically with the continuous 
time model. We used for this comparison the analytical result for the 
martingale defect in Eq.~(8.25) in Chapter 8.4 of \cite{Lewis2017}. 

The asymptotic martingale property implies the following result, which will be used in the
$n\to \infty$ option asymptotics.

\begin{corollary}\label{lemma:2}
For any $\varrho \leq 0$ we have 
\begin{equation}\label{infmartingale}
\inf_{u,v>0} \left\{ \frac12 I(u,v) + a \varrho^2 u - \varrho 2\sqrt{2a} (v-1)
\right\} = 0 \,,
\end{equation}
where $a:=4\beta\rho^{2}$.
Denote the point where the infimum is reached as $(u_m,v_m)$. This
point depends on the product $\varrho\sqrt{a}$ and approaches $(1,1)$ 
as $\varrho\sqrt{a} \to 0$.
\end{corollary}

\begin{proof}
The asymptotic martingale property (\ref{S0m}) implies
\begin{equation}
\lim_{n\rightarrow\infty}\mathbb{E}\left[e^{- \frac12 \varrho^2 V_n + \frac{\varrho}{\omega} 
(\sigma_n-\sigma_0)}\right]
= \lim_{n\rightarrow\infty}\mathbb{E}\left[e^{nF(\frac{V_n}{n}, \frac{1}{\omega n} \sigma_n)}\right] = 1\,,
\end{equation}
with 
\begin{equation}
F(x,y) = - \frac12 \varrho^2 x + \varrho (y - v_0)\,.
\end{equation}
By Varadhan's lemma we have
\begin{equation}
\lim_{n\to \infty} \frac{1}{n} \log \mathbb{E}\left[e^{nF(\frac{V_n}{n}, \frac{1}{\omega n} \sigma_n)}\right]
= -\inf_{u,v} \left\{ \mathcal{I}(u,v) - F(u,v) \right\} = 0\,.
\end{equation}
This is written equivalently as (\ref{infmartingale}).

\end{proof}

\subsection{Large deviations for $\mathbb{P}(\frac{1}{n} \log S_n \in \cdot )$}


We are now in a position to prove the large deviations property for 
$\mathbb{P}(\frac{1}{n} \log S_n\in\cdot)$ in the correlated log-normal 
SABR model.

\begin{proposition}\label{prop:LDlogS}
Consider the $n\to \infty$ limit at fixed $\beta = \frac12 \omega^2 \tau n^2$
and $\rho = \sigma_0 \sqrt{\tau}$.
In this limit $\mathbb{P}(\frac{1}{n} \log S_n \in \cdot )$ satisfies a LDP 
with rate function
\begin{align}\label{I1rf}
\mathcal{I}_X(x,\varrho) &= \inf_{x = \varrho_\perp \sqrt{u} z - \frac12 u + \varrho (v-v_0)}
\left( \mathcal{I}(u,v) + \frac12 z^2 \right) \\
&=
\inf_{(u,v) \in \mathbb{R}_+^2} \left( \mathcal{I}(u,v) + 
\frac{1}{2\varrho_\perp^2 u} \left(x + \frac12 u - \varrho (v-v_0)\right)^2 \right)\,, \nonumber
\end{align}
for $x \geq 0$, and $\mathcal{I}_X(x,\varrho)=\infty$ otherwise.
We denoted here $v_0 = \frac{\rho}{\sqrt{2\beta}}$.
\end{proposition}

\begin{proof}
From Proposition~\ref{prop:I} we know that
$\mathbb{P}((\frac{1}{n} V_n, \frac{1}{n\omega} \sigma_{n}) \in \cdot )$
satisfies a LDP with rate function $\mathcal{I}(u,v)$.
We also have that $\mathbb{P}(\frac{1}{\sqrt{n}} Z \in \cdot)$ 
satisfies a LDP with rate function 
\begin{equation}\label{J2sol}
J_2(x) = \frac12 x^2\,.
\end{equation}
This follows from the G\"artner-Ellis theorem (see e.g. \cite{Dembo}) by noting that
for any $\theta\in\mathbb{R}$ we have
$\mathbb{E}\left[e^{\theta n \frac{Z}{\sqrt{n}}}\right]= e^{\frac12\theta^2 n}$
so that $\Lambda(\theta) := \lim_{n\to \infty} \frac{1}{n} \log 
\mathbb{E}\left[e^{\theta n \frac{Z}{\sqrt{n}}}\right] =
\frac12\theta^2$,
which implies that
\begin{equation}
J_2(x) = \sup_{\theta\in\mathbb{R}}\{\theta x - \Lambda(\theta)\} = \frac12 x^2\,.
\end{equation}
Writing (\ref{scheme2}) as
\begin{equation}
\frac{1}{n} \log S_n = - \frac12 \frac{V_n}{n} + \varrho_\perp \sqrt{\frac{V_n}{n}} \cdot 
\frac{Z}{\sqrt{n}} + \varrho \frac{1}{\omega n} (\sigma_n - \sigma_0)\,,
\end{equation}
we get from the contraction principle (see e.g. \cite{Dembo}) that 
$\mathbb{P}(\frac{1}{n} \log S_n \in \cdot)$ satisfies a LDP 
with rate function 
\begin{equation}
\mathcal{I}_X(x,\varrho) = \inf_{-\frac12 u + \varrho_\perp \sqrt{u} z + \varrho (v-v_0)= x} \{ \mathcal{I}(u,v) + J_2(z) \}\,.
\end{equation}
This completes the proof of (\ref{I1rf}). 
\end{proof}

\subsection{Properties of the rate function $\mathcal{I}_X(x,\varrho)$}

We give here a few properties of the rate function $\mathcal{I}_X(x,\varrho)$
introduced in the previous section.

\begin{proposition}\label{prop:Iprop}
The rate function $\mathcal{I}_X(x,\varrho)$ 
vanishes for $x_L =-\frac12\rho^{2}$. That is,
\begin{equation}\label{zero}
\mathcal{I}_X\left( x_L, \varrho \right) = 0\,.
\end{equation}
\end{proposition}

\begin{proof}
Eq.~(\ref{zero}) follows by noting that for $x = - \frac12\rho^2$,
the minimizer in the extremal problem of Prop.~\ref{prop:LDlogS}
is reached at $u=\rho^2, v = v_0$. At this point the rate function vanishes.
\end{proof}


\begin{remark}
The result (\ref{zero}) agrees with the almost sure limit for
$\frac{1}{n} \log S_n$ in Proposition~\ref{prop:LLN}. 
\end{remark}

The rate function $\mathcal{I}_X(x,\varrho)$ has a scaling property and depends
only on $x/\rho^2$ and the product $a:= 4\beta \rho^2=
2\sigma_0^2\omega^2 (\tau n)^2$ 
\begin{equation}\label{JXdef}
\mathcal{I}_X(x,\varrho) = \frac{1}{8\beta} 
\mathcal{J}_X\left(x/\rho^{2}; 4\beta\rho^2, \varrho\right) \,,
\end{equation}
where
\begin{equation}\label{JXinf}
\mathcal{J}_X(y;a,\varrho) = \inf_{u,v} \left\{
\frac12 I(u,v) + \frac{a}{\varrho_\perp^2 u} 
\left( y + \frac12 u - \varrho (v-1) \sqrt{\frac{2}{a}} \right)^2 
\right\}\,.
\end{equation}


Expressed in terms of this function, the property (\ref{zero}) 
reads $\mathcal{J}_X\left(-\frac12;a,\varrho\right) = 0$.
The rate function has a calculable expansion around its minimum
given by the following result.

\begin{proposition}
The leading term in the expansion of the rate function $J_X(y;a,\varrho)$ 
around its minimum at $y=-\frac12$ is
\begin{equation}\label{JXexp}
J_X(y;a,\varrho) = \frac{6a}{6+a-3\sqrt{2a} \varrho} \left(y+\frac12\right)^2 + 
O\left(\left(y+\frac12\right)^3\right)\,.
\end{equation}
\end{proposition}

\begin{proof}
The minimum condition is
\begin{equation}
u \partial_u \Lambda(u,v) = 0\,,\qquad 
v \partial_v \Lambda(u,v) = 0\,.
\end{equation}

The minimizer in the extremal problem (\ref{JXinf}) 
for this rate function can be expanded in powers of $y+\frac12$:
\begin{eqnarray}
&& x_* = \log u_* = a_1 \left(y+\frac12\right) + a_2 \left(y+\frac12\right)^2 + \cdots\,, \\
&& y_* = \log v_* = b_1 \left(y+\frac12\right) + b_2 \left(y+\frac12\right)^2 + \cdots\,. 
\end{eqnarray}
Substituting the expansion of the rate function $I(u,v)$ in Proposition~\ref{prop:Iuvexp} 
gives a sequence of equations for the coefficients $a_i,b_i$. The first coefficients are
\begin{equation}
a_1 = \frac{-2a + 3\sqrt{2a} \varrho}{6 - 3\sqrt{2a} \varrho + a}\,,\quad
b_1 = -\frac{3(a - 2\sqrt{2a} \varrho)}{2(6 - 3\sqrt{2a} \varrho + a)}\,.
\end{equation}

Substituting the expansion into $J_X(y;a,\varrho)$ gives an expansion in
$y+\frac12$ with coefficients expressed in terms of $a_i,b_i$. The leading term
is given in (\ref{JXexp}).

\end{proof}

We prove next a lower bound on the rate function, and an equality on its 
value at a certain point, which will play an important 
role in the $n\to \infty$ asymptotics of the option prices.

\begin{proposition}
Assume $\varrho \leq 0$. The rate function $J_X(y;a,\varrho)$ is bounded from below as
\begin{equation}\label{JXlowbound}
J_X(y; a, \varrho) \geq 2 ay\,.
\end{equation}
The lower bound is reached at 
\begin{equation}\label{yRdef}
y_R = \frac12 \left(1-2\varrho^2\right) u_m + \varrho \sqrt{\frac{2}{a}} (v_m-1)\,,
\end{equation}
where $(u_m,v_m)$ are given by Corollary \ref{lemma:2}. At this point we have
\begin{equation}\label{zeroR}
J_X(y_R; a, \varrho) = 2 ay_R \,. 
\end{equation}
\end{proposition}

\begin{proof}
By Corollary \ref{lemma:2}, we have a lower bound on $I(u,v)$ 
\begin{equation}
\frac12 I(u,v) \geq - a \varrho^2 u + \varrho 2\sqrt{2a} (v-1)\,,\quad u,v > 0\,.
\end{equation}

Substituting into the expression (\ref{JXdef}) gives the lower bound
\begin{equation}
J_X(y;a,\varrho) - 2 a y \geq \inf_{u,v>0} \left\{ \frac{a}{\varrho_\perp^2 u} (y - y_R)^2 \right\}\,,
\end{equation}
with $y_R$ defined in (\ref{yRdef}). This proves the lower bound (\ref{JXlowbound}).

In order to prove the equality in (\ref{JXlowbound})
for $y=y_R$, note that by Corollary \ref{lemma:2} there exist $(u_m,v_m)$ such that
the lower bound on $I(u,v)$ above is reached. Substituting into $J_X(y_R;a,\varrho)$
we get that this is equal to $2ay_R$, as stated. Expressed in terms of the
rate function $\mathcal{I}_X(k,\varrho)$, the relation (\ref{zeroR}) reads
$\mathcal{I}_X(k_R,\varrho) = k_R$ with $k_R = y_R \rho^2$.
\end{proof}



In the uncorrelated case $\varrho=0$, the rate function 
$\mathcal{I}_X(x,\varrho)$ simplifies further. 
For this case the extremal problem (\ref{JXinf}) can be solved in closed form,
using the result for 
the rate function $\mathcal{J}_{\rm BS}(x)$ obtained in \cite{PZSIFIN}. 
The result for this rate function is given in Corollary~\ref{corr:JXrho0}
in Appendix \ref{sec:a1}.
When $\varrho=0$, the rate function $\mathcal{I}_X(x,0)$ 
satisfies the symmetry relation
\begin{equation}\label{Isymm}
\mathcal{I}_X(x,0) - \mathcal{I}_X(-x,0) = x\,,
\end{equation}
see Proposition \ref{prop:Isymm}
in Appendix \ref{sec:a1}. Expressed in terms of $\mathcal{J}_X(x ; a, 0)$
this reads
\begin{eqnarray}
\label{JXsymm2}
&& \mathcal{J}_X(x ; a, 0) - \mathcal{J}_X(-x ; a, 0) = 2ax\,.
\end{eqnarray}


\section{Option price and implied volatility asymptotics}
\label{sec:5}

We derive in this section option prices asymptotics in the time discretized
log-normal SABR model discretized in time under the scheme (\ref{scheme2}).
This result will be used to obtain the asymptotics of the implied volatility.

\subsection{Option prices asymptotics}
\label{sec:4.1}

We consider the vanilla European call and put options:
\begin{equation}
C(n):=\mathbb{E}\left[\left(S_{n}-K\right)^{+}\right],
\qquad
P(n):=\mathbb{E}\left[\left(K-S_{n}\right)^{+}\right],
\end{equation}
where $K>0$ is the strike price, and we write $C(n)$ and $P(n)$ to emphasize
the dependence on the number of time steps $n$. We study here the $n\to
\infty$ asymptotics of the option prices with strike $K = S_0 e^{nk}$.
The asymptotics will be shown to be different in the three regimes:
\begin{enumerate}
\item The large-strike regime $k> y_R \rho^2$. In  this regime the
call option is out-of-the-money (OTM) and $\lim_{n\to \infty} C(n)=0$;
\item The intermediate strike regime $- \frac12 \rho^2 \leq k\leq y_R\rho^2$;
In  this regime the
covered call option is OTM and $\lim_{n\to \infty} (S_0-C(n))=0$;
\item The small-strike regime $k < -\frac12\rho^2$.
In  this regime the
put option is OTM and $\lim_{n\to \infty} P(n)=0$;
\end{enumerate}
Here $y_R$ is given by Eq.~(\ref{yRdef}).
The asymptotics of the option prices are given by the following result.

\begin{theorem}\label{Thm:1} 
The $n\to \infty$ asymptotics of the option prices are given by
\begin{equation}
k - \mathcal{I}_X(k,\varrho)
=\begin{cases}
\lim_{n\rightarrow\infty}\frac{1}{n}\log C(n) &
\text{for $k> y_R \rho^{2}$},
\\
\lim_{n\rightarrow\infty}\frac{1}{n}\log(S_{0}-C(n)) &
\text{for $-\frac{1}{2}\rho^{2}\leq k\leq y_R \rho^{2}$},
\\
\lim_{n\rightarrow\infty}\frac{1}{n}\log P(n) &
\text{for $k<-\frac{1}{2}\rho^{2}$},
\end{cases}
\end{equation}
where $\mathcal{I}_X(k, \varrho)$ is the rate function given by 
Proposition~\ref{prop:LDlogS}. 
\end{theorem}

\begin{proof}
Conditioning on $(V_n,\sigma_n)$, the asset price $S_n$ is log-normally 
distributed and can be written as (with $Z\sim N(0,1)$ independent of $\sigma_n,V_n)$)
\begin{equation}\label{Fndef}
S_n=F_n e^{\varrho_\perp \sqrt{V_n} Z - \frac12 \varrho_\perp^2 V_n}\,,\quad
F_n =S_0 e^{-\frac12\varrho^2 V_n + \frac{\varrho}{\omega}(\sigma_n-\sigma_0)}\,.
\end{equation}
The option prices can be written as expectations over $(V_n,\sigma_n)$ of the
Black-Scholes formula
\begin{align}
&C(n) = \mathbb{E}[F_n N(d_1)] - K \mathbb{E}[N(d_2)]\,, \\
&P(n) = K \mathbb{E}[N(-d_2)] - \mathbb{E}[F_n N(-d_1)]\,,
\end{align}
where $d_{1,2}$ are random variables 
\begin{align}
d_{1,2} &:= \frac{1}{\varrho_\perp \sqrt{V_n}} \log \frac{F_n}{K}
\pm \frac12 \varrho_\perp \sqrt{V_n} \\
&=\frac{1}{\varrho_\perp \sqrt{V_n}} \left( \log\frac{S_0}{K}
- \frac12 \varrho^2 V_n + \varrho \frac{1}{\omega} (\sigma_n -\sigma_0)
\pm \frac12 \varrho_\perp^2 V_n \right) \,.\nonumber
\end{align}

We are interested in the $n\to \infty$ asymptotics of the option prices for
strike $K  = S_0 e^{n k}$. 

We give the proof of the $n\to \infty$ asymptotics for the OTM call option;
the other two cases follow analogously. The proof follows by upper and lower 
bounds on $C(n)$.

(i) Neglecting the second term in $C(n)$ gives the upper bound
\begin{equation}
C(n) \leq \mathbb{E}[F_n N(d_1)] = S_0 \mathbb{E}\left[
e^{-\frac12\varrho^2 V_n + \varrho \frac{1}{\omega} (\sigma_n-\sigma_0)}
N(d_1)\right] \,.
\end{equation}

Using $N(d_{1,2} ) = \mathbb{P}(d_{1,2} > Z)$ with $Z\sim N(0,1)$ independent
of $\sigma_n,V_n$, we have 
\begin{align}\label{ICup}
\limsup_{n\to \infty} \frac{1}{n} \log C(n) 
&\leq
\limsup_{n\to \infty} \frac{1}{n} \log 
\mathbb{E}\left[e^{nF(\frac{V_n}{n},\frac{\sigma_n}{\omega n})}
\mathbb{P}\left(\frac{d_1}{\sqrt{n}} > \frac{Z}{\sqrt{n}}\right)\right] \\
&=
- \inf_{D_1(u,v) \geq z } 
\left( \mathcal{I}(u,v) + \frac12 z^2 + \frac12\varrho^2 u - \varrho (v-v_0)\right) 
:= \mathcal{I}_C^{\rm up}(k)\,, \nonumber
\end{align}
where we used in the last step Varadhan's lemma for the expectation
containing $F(x,y) := -\frac12\varrho^2 x + \varrho (y-v_0)$. The constraint
is defined in terms of
\begin{eqnarray}\label{D1def}
&& D_1(u,v) := \frac{1}{\varrho_\perp \sqrt{u}}\left(
- k + \frac12 (1 - 2\varrho^2) u + \varrho(v - v_0) \right)\,, \\
\label{D2def}
&& D_2(u,v) := \frac{1}{\varrho_\perp \sqrt{u}}\left(
- k - \frac12 u + \varrho(v - v_0) \right) \,.
\end{eqnarray}
They satisfy $D_1(u,v) = D_2(u,v) + \varrho_\perp \sqrt{u}$.

(ii) We prove also a matching lower bound. For any $\epsilon > 0$ we have
\begin{eqnarray}
&& C(n) = \mathbb{E}\left[\left(S_n - S_0 e^{nk}\right) 1_{S_n \geq S_0 e^{nk}}\right] \\
&& \qquad \geq \mathbb{E}\left[\left(S_n - S_0 e^{nk}\right) 1_{S_n \geq S_0 e^{nk+n \epsilon}}\right] \nonumber \\
&& \qquad \geq (e^{n\epsilon} - 1) S_0 e^{nk} \mathbb{P}\left(S_n \geq S_0 e^{nk + n\epsilon}\right)\nonumber \\
&& \qquad = (e^{n\epsilon} - 1) S_0 e^{nk} 
\mathbb{P}\left( \frac{d_2}{\sqrt{n}} \geq \frac{Z}{\sqrt{n}}\right)\,.\nonumber
\end{eqnarray}
This gives
\begin{align}
\liminf_{n\to \infty} \frac{1}{n} \log C(n) 
&\geq
k + \epsilon + \lim_{n\to \infty} \frac{1}{n} \log 
\mathbb{P}\left(\frac{d_2}{\sqrt{n}} \geq \frac{Z}{\sqrt{n}}\right) 
\\
&= k + \epsilon - \inf_{D_2(u,v) \geq z } 
\left( \mathcal{I}(u,v) + \frac12 z^2 \right) \,.\nonumber
\end{align}
Since this inequality holds for any $\epsilon >0$, we get
\begin{eqnarray}\label{IClow}
&& \liminf_{n\to \infty} \frac{1}{n} \log C(n)  \geq - \mathcal{I}_C^{\rm low}(k)
:= k - \inf_{D_2(u,v) \geq z } 
\left( \mathcal{I}(u,v) + \frac12 z^2 \right) \,.
\end{eqnarray}

The bounds have different behavior depending on $k$, as 
\begin{equation}\label{IC2cases}
\mathcal{I}_C^{\rm up}(k) = 
\begin{cases}
0\,, &   k < y_R\rho^2, \\
\mathcal{I}_X(k,\varrho) - k\,, &  k > y_R\rho^2,
\end{cases}
\qquad
\mathcal{I}_C^{\rm low}(k) = 
\begin{cases}
0\,, &   k < -\frac12 \rho^2, 
\\
\mathcal{I}_X(k,\varrho) -k\,, &  k > -\frac12\rho^2. 
\end{cases}
\end{equation}

This follows from a study of the global minimum of the functions in the
constrained extremal problems for the bounds in 
Eq.~(\ref{ICup}) and (\ref{IClow}) in relation to the constraints. 

Consider first the upper bound. By Corollary~\ref{lemma:2} the function in 
Eq.~(\ref{ICup}) has a global minimum equal to
zero at $(u_m\rho^2,v_m\rho^2)$ and $z=0$. For sufficiently small $k<y_R\rho^2$, 
this point is within the region allowed by the constraint $D_1(u,v)\geq z$. 
This proves the upper line in Eq.~(\ref{IC2cases}).
For $k>y_R\rho^2$, the global minimum is excluded by the condition 
$D_1(u,v)\geq z$.
Convexity of $\mathcal{I}(u,v)$ implies that any local minimum in 
Eq.~(\ref{ICup}) is also the global minimum, so the inf is reached on the
boundary of the region $D_1(u,v)= z$, but not in the interior of the region.

The function appearing in the lower bound
$\mathcal{I}_C^{\rm low}(k)$ has a global minimum of zero at 
$u=\rho^2, v=v_0,z=0$, which is allowed by the constraint $D_2(u,v)\geq z$
only for $k < -\frac12\rho^2$. For $k>-\frac12\rho^2$ the infimum in 
Eq.~(\ref{IClow}) is reached on the boundary of the
region $D_2(u,v)= z$. On the respective boundaries, the two bounds coincide
\begin{align}
- \inf_{D_1(u,v) = z } 
\left( \mathcal{I}(u,v) + \frac12 z^2 +\frac12 \varrho^2 u - \varrho (v-v_0) \right) 
&=
k - \inf_{D_2(u,v) = z } 
\left( \mathcal{I}(u,v) + \frac12 z^2 \right) 
\\
&= k - \mathcal{I}_X(k,\varrho)\,.\nonumber
\end{align}
This proves the lower line equations in Eq.~(\ref{IC2cases}).
This completes the proof of the result for the OTM call. The proofs for 
the other two cases are similar.
\end{proof}

\subsection{Implied volatility}
\label{sec:4.2}

Using the option prices asymptotics of Theorem~\ref{Thm:1} one can obtain
the asymptotics of the implied volatility in the
log-normal SABR model under the discretization scheme (\ref{scheme2}).

\begin{theorem}\label{thm:main}
Consider the SABR model with correlation $\varrho \leq 0$ 
discretized in time with $n$ points under the scheme (\ref{scheme2}).
In the limit $n \to \infty$ at fixed $\rho^2 = \sigma_0^2 \tau,
\beta = \frac12 \omega^2 \tau n^2$, the implied volatility for 
maturity $T := t_n$ and log-strike $x := \log(K/S_0)$ is given by
\begin{equation}\label{SigBSsol}
\sigma_{\rm BS}(x,T ) = \sigma_0 
\Sigma_{\rm BS}\left( \frac{x}{\sigma_0^2 T};a \right)\,,
\end{equation}
where the equality in \eqref{SigBSsol} means the LHS/RHS goes to one in the 
limit, and
\begin{equation}\label{SigBSdef}
\Sigma_{\rm BS}(y;a) 
:= 
\begin{cases}
\left|\sqrt{\frac{1}{a} \mathcal{J}_X(y;a,\varrho) - 2y} - 
      \sqrt{\frac{1}{a} \mathcal{J}_X(y;a,\varrho)}\right|
&\text{for $y > y_R$ and $y<-\frac12$}, 
\\
\sqrt{\frac{1}{a} \mathcal{J}_X(y;a,\varrho) - 2y} + 
\sqrt{\frac{1}{a} \mathcal{J}_X(y;a,\varrho)}
&\text{for $-\frac12 \leq y \leq y_R$},
\end{cases}
\end{equation}
where $\mathcal{J}_X(y;a,\varrho)$ is the rate function defined in (\ref{JXdef}) 
and $a := 2(\sigma_0^2 T)(\omega^2 T)$. 
\end{theorem}

\begin{proof}
The result is a standard transfer relation of the option price asymptotics to
implied volatility. Similar results are obtained in Corollary 2.14 of the Forde and
Jacquier paper \cite{FJ1} for the Heston model. A more general treatment 
of these transfer results is given in Gao and Lee \cite{GaoLee}. The
different cases of the option price asymptotics in regions (1) and (3) of
Theorem~\ref{Thm:1} correspond to the regime (+) in Section 4.1 of \cite{GaoLee},
and the region (2) corresponds to the regime (-). 

The $n\to \infty$ limit of the implied volatility for maturity $t_n = n\tau$
and log-strike limit $x = \log(K/S_0) = nk$ with constant $k$ is
\begin{align}\label{impsol}
&\lim_{n\to \infty} \sigma_{\rm BS}^2(x,t_n)\tau 
\\
&=
\begin{cases}
2\left( 2 \mathcal{I}_X(k, \varrho) - k) - 
4\sqrt{\mathcal{I}_X(k, \varrho) (\mathcal{I}_X(k, \varrho) - k ) } \right) &
\text{for $k > k_R$ and $k< k_L$} \\
2\left( 2 \mathcal{I}_X(k, \varrho) - k) +
4\sqrt{\mathcal{I}_X(k, \varrho) (\mathcal{I}_X(k, \varrho) - k) } \right) & 
\text{for $k_L \leq k \leq k_R$} 
\end{cases}
\nonumber
\\
&= 
\begin{cases}
2 \left( \sqrt{ \mathcal{I}_X(k,\varrho)} - \sqrt{ \mathcal{I}_X(k, \varrho) - k} \right)^2
& \text{for $k > k_R$ and $k < k_L$} \\
2 \left( \sqrt{ \mathcal{I}_X(k, \varrho) } + \sqrt{ \mathcal{I}_X(k, \varrho) -k } \right)^2
& \text{for $k_L \leq k \leq k_R$}
\end{cases}\,,
\nonumber
\end{align}
where $\mathcal{I}_X(k, \varrho)$ is the rate function given by 
Proposition~\ref{prop:LDlogS}
and $k_L= -\frac12\rho^2\,, k_R = y_R \rho^2$.

We note that although the result was derived in discrete time, 
the asymptotic implied volatility does not depend on the
time step $\tau$, but depends only on the product $T = \tau n = t_n$. The
result (\ref{impsol}) can be written equivalently as (\ref{SigBSsol}). 
\end{proof}

We note that the asymptotic implied volatility of Theorem~\ref{thm:main}
has a scaling property, as it depends only on the two variables 
$\frac{x}{\sigma_0^2 T}$ and $a = 2(\sigma_0^2 T)(\omega^2 T)$.

\begin{figure}[b!]
\begin{center}
\includegraphics[height=50mm]{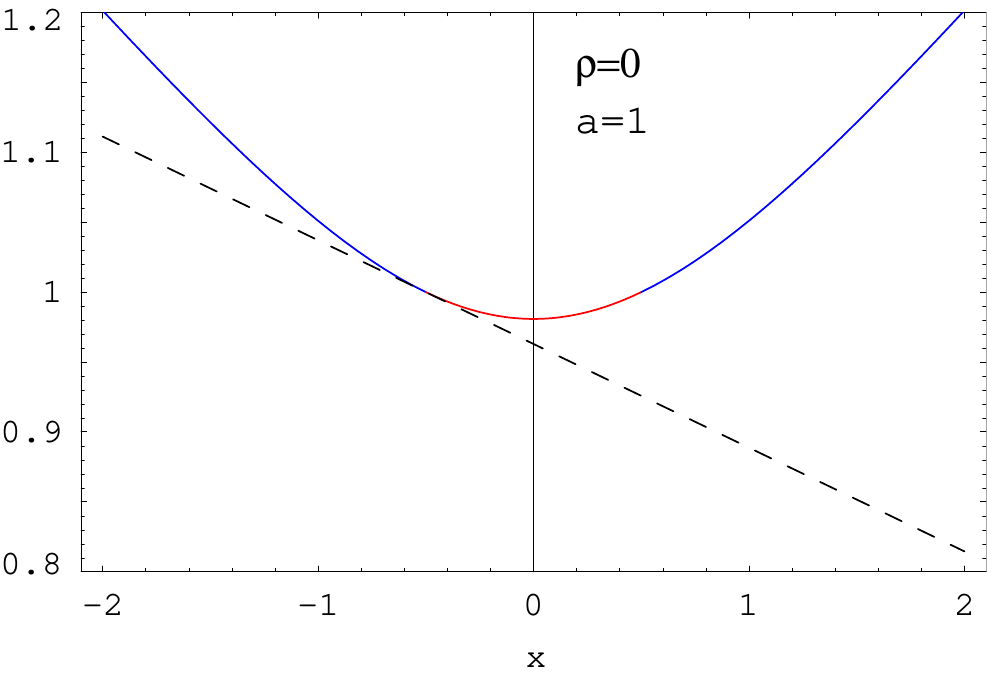}
\includegraphics[height=50mm]{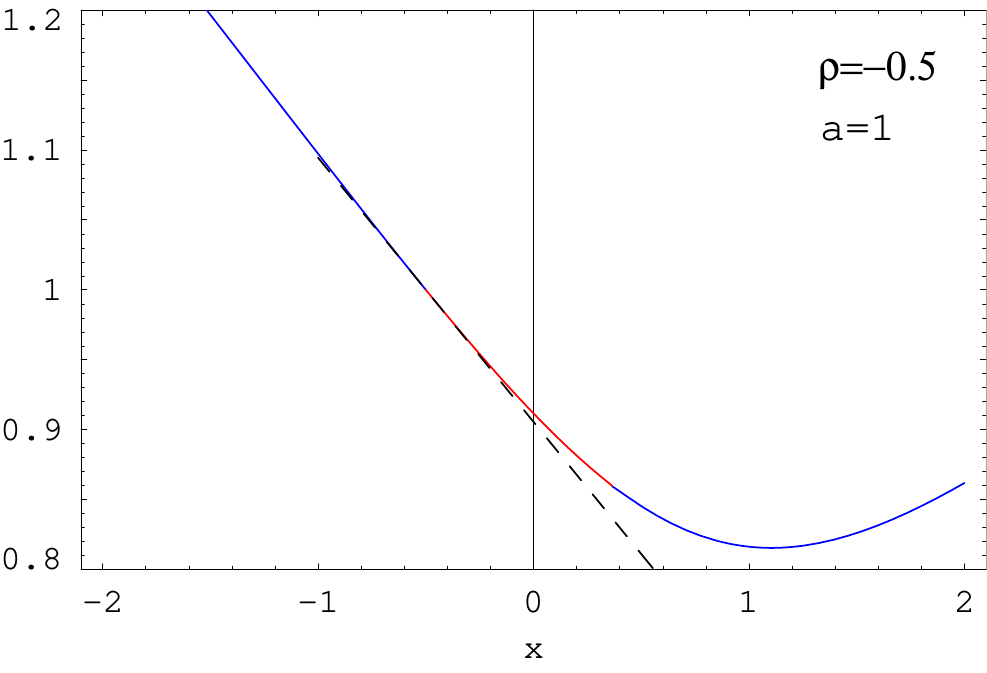}
\end{center}
\caption{
Plot of the asymptotic implied volatility $\Sigma_{\rm BS}(x;a,\varrho)$ 
in the log-normal discrete-time SABR model for parameters $a=1.0$ and
$\varrho = 0.0$ (left panel) and $\varrho = -0.5$ (right panel). The
two branches in Theorem~\ref{thm:main} are shown as the blue and red curves.
The dashed line is the linear approximation (\ref{SigBSlin}).}
\label{Fig:1}
\end{figure}

\begin{remark}
The result of Theorem~\ref{thm:main} reveals the existence of three
regions of log-strike separated by $x_L=-\frac12\sigma_0^2 T$
and $x_R =  y_R \sigma_0^2 T$.
At the switch points we have 
\begin{equation}
\sigma_{\rm BS}\left( x_L, T\right) = \sigma_0 \,,\quad
\sigma_{\rm BS}\left( x_R, T\right) = \sqrt{2y_R} \sigma_0\,.
\end{equation}
This is illustrated in Figure~\ref{Fig:1} which shows the implied volatility
function $\Sigma_{\rm BS}(y;a)$ for $a=1.0$ and correlation
$\varrho = 0$ (left) and $\varrho=-0.5$ (right). The three regions of
Theorem~\ref{thm:main} are shown with different colors.
This is different from the SABR formula (\ref{Hagan}) which does not 
distinguish between these regions. 
\end{remark}

\begin{remark}
The result of Theorem~\ref{thm:main}
is similar to the large maturity asymptotics for the Heston model
derived by Forde and Jacquier \cite{FJ1}. However we note
also a difference. In their result for the Heston model, the asymptotic 
implied volatility does not depend on $\sigma_0$, the initial condition for 
the volatility. This is because their
rate function does not depend on $\sigma_0$. On the other hand, with
our scaling $\sigma_0$ appears through the scaling variable $\rho$, 
which introduces dependence on $\sigma_0$ in the asymptotic implied
volatility.
\end{remark}

\begin{remark}\label{remark:symm}
For zero correlation $\varrho=0$, the implied volatility given by 
Theorem~\ref{thm:main} is symmetric in log-strike (see the left panel
in Fig.~\ref{Fig:1} for an illustration)
\begin{equation}
\sigma_{\rm BS}(-x,T) = \sigma_{\rm BS}(x,T)\,.
\end{equation}
This follows from the symmetry relation (\ref{Isymm}) for the rate function 
in the zero correlation limit $\mathcal{I}_X(k,0) - \mathcal{I}_X(-k,0) = k$.
This agrees with the well-known result of \cite{Touzi} that the implied 
volatility in an uncorrelated stochastic volatility model is a symmetric 
function of log-strike.
\end{remark}

The leading quadratic term in the expansion of the rate function around
its minimum at $y=-\frac12$ gives a linear
approximation for $\Sigma_{BS}(y;a)$ around $y=-\frac12$
\begin{equation}\label{SigBSlin}
\Sigma_{BS}(y;a) = 1 - (1-\sqrt{c}) \left(y+\frac12\right)
+ O\left(\left(y+\frac12\right)^2\right)\,, 
\end{equation}
with $c = \frac{1}{1 + \frac16 a - \sqrt{\frac{a}{2}}\varrho}$.
This linear approximation is shown as the dashed line in Fig.~\ref{Fig:1}.

\section{Limiting cases}
\label{sec:6}

We study in this section the limits of the asymptotic implied volatility
of Theorem~\ref{thm:main} in several regimes of short maturity $T\to 0$
and extreme strikes $|x|\to \infty$ at fixed maturity. 
The results are compared with existing results in the literature.

\subsection{Short maturity limit}
\label{sec:6.1}

Consider the $T\to 0$ limit of the implied volatility. 
This limit is obtained by taking $a\to 0$ at fixed $\sqrt{a} y = \sqrt2 \zeta$,
with $\zeta= \frac{\omega}{\sigma_0} x$. Assuming that the limit exists,
define
\begin{equation}\label{Jdef}
J(\zeta;\varrho) := \lim_{a\to 0, \sqrt{a} y =\sqrt2 \zeta} 
\mathcal{J}_X(y;a,\varrho)\,.
\end{equation}

Let us take this limit in the rate function $\mathcal{J}_X(y;a,\varrho)$, 
expressed as the extremal problem (\ref{JXinf}). This gives
\begin{equation}\label{Jzinf}
J(\zeta;\varrho) = 
\inf_{u,v>0} \left\{ \frac12 I(u,v) + 
\frac{2}{\varrho_\perp^2 u} \left( \zeta  
- \varrho (v-1) \right)^2 \right\} \,.
\end{equation}

For $\zeta=0$ the minimizer is $u_*=1,v_*=1$; at this point 
the rate function vanishes $J(0;\varrho) = 0$.

\begin{proposition}
The solution of the extremal problem (\ref{Jzinf}) for the rate
function $J(\zeta;\varrho)$ has the expansion around $\zeta=0$
\begin{equation}\label{Jzexp}
J(\zeta;\varrho) =
2\zeta^2 - 2 \varrho \zeta^3 + \left(-\frac23 + \frac52 \varrho^2\right) 
\zeta^4 + O(\zeta^5) \,.
\end{equation}

In the uncorrelated case $\varrho=0$ the extremal
problem (\ref{Jzinf}) can be solved exactly 
\begin{equation}\label{Jzero}
J(\zeta;0) = 
2\log^2 \left( \sqrt{\zeta^2+1} + |\zeta|\right)\,.
\end{equation}
\end{proposition}

\begin{proof}
The infimum condition in (\ref{Jzinf}) can be expressed as the
vanishing of the partial derivatives of the function of $(u,v)$ on the
right hand side. This gives two equations for the extremal point
$(u_*,v_*)$
\begin{align}\label{dIuv}
& u \partial_u I(u,v) = \frac{4}{\varrho^2_\perp u}
\left(\zeta -\varrho (v-1)\right)^2\,,\\
& v \partial_v I(u,v) = \frac{8\varrho}{\varrho^2_\perp u}
v \left(\zeta -\varrho (v-1)\right) \,.
\nonumber
\end{align}
Introducing $\epsilon = \log u_*, \eta = \log v_*$, the
minimizers can be expanded in $\zeta$ as
\begin{equation}
\epsilon = a_1 \zeta + a_2 \zeta^2 + O(\zeta^3)\,,\quad
\eta = b_1 \zeta + b_2 \zeta^2 + O(\zeta^3) \,.
\end{equation}
Using the expansion for the rate function $I(u,v)$, 
the two equations in (\ref{dIuv}) can be expanded also in $\zeta$. Requiring
the equality of the terms of each order in $\zeta$ gives successive
equations for $a_i,b_j$ which can be solved recursively. The first two
coefficients are
\begin{equation}
a_1=b_1=\varrho\,,\quad 
a_2 = \frac23 - \varrho^2 \,,\quad
b_2 = \frac12(1-2\varrho^2) \,.
\end{equation}
The coefficients $a_{1,2},b_{1,2}$ are sufficient to 
determine the expansion of the rate function $J(\zeta;\varrho)$
to order $O(\zeta^2)$, with the result quoted in (\ref{Jzexp}).

We give next the proof of (\ref{Jzero}) for the uncorrelated case.
As $a\to 0$ at fixed $ay^2$
we have $|y|\to \infty$, such that we use the $y>y_R$ branch
of the function $\Sigma_{\rm BS}(y;a)$ in Theorem~\ref{thm:main}. 
The rate function $\mathcal{J}_X(y;a,0)$ is given by Corollary~\ref{corr:JXrho0}.
The equation for $\xi$ in this result becomes in this limit
\begin{equation}
\sinh^2 \left(\frac{\xi}{2}\right) = \frac12 ay^2 = \zeta^2\,,
\end{equation}
which determines $\xi$ up to a sign as
\begin{equation}
\xi = \pm \log\left(1 + 2\zeta^2 + 2\sqrt{\zeta^2(1+\zeta^2)}\right) = 
\pm 2 \log\left(\sqrt{\zeta^2+1} + \sqrt{\zeta^2}\right)\,.
\end{equation}
The rate function is
\begin{equation}
\lim_{a\to 0, ay^2 = 2\zeta^2} \mathcal{J}_X(y;a,0) = \frac12 \xi^2 
=\xi \tanh(\xi/2) + 2\zeta^2 \frac{\xi}{\sinh \xi} = 
2\log^2\left(\sqrt{\zeta^2+1} + \sqrt{\zeta^2}\right)\,.
\end{equation}
\end{proof}

\begin{remark}
The expansion (\ref{Jzexp}) agrees with the first three terms in the
Taylor expansion of the function 
\begin{eqnarray}\label{Jguess} 
\tilde J(\zeta,\varrho):= 
2\log^2 
\frac{\sqrt{1 + 2\varrho \zeta + \zeta^2} + \zeta + \varrho}{1+\varrho}\,.
\end{eqnarray}
\end{remark}

Numerical testing shows that the rate function $J(\zeta;\varrho)$ is
reproduced to very good precision by this function; however we could not
prove their equality analytically, except for the uncorrelated case $\varrho=0$,
when (\ref{Jguess}) reduces to (\ref{Jzero}).

The asymptotic implied volatility in the small-maturity limit can be expressed
in terms of the rate function $J(\zeta;\varrho)$ given by the limit (\ref{Jdef}).

\begin{proposition}
Assume that the limit (\ref{Jdef}) exists and is given by the rate function
$J(\zeta;\varrho)$. 
Then the implied volatility in the $T\to 0$ limit of the SABR model with 
correlation
$\varrho$ and $\omega\to 0,\sigma_0 \to \infty$ at fixed $\sigma_0\omega$ is
\begin{equation}\label{limSigBS}
\lim_{\sigma_0 \to \infty} \frac{1}{\sigma_0}
\sigma_{\rm BS}(x,T) = \frac{\sqrt{2} \zeta}{\sqrt{J(\zeta;\varrho)}},
\qquad\text{with $\zeta = \frac{\omega}{\sigma_0} x$}.
\end{equation}
\end{proposition}

\begin{proof}
Start with the result for $\Sigma_{\rm BS}(y;a)$ 
from Theorem \ref{Thm:1}
\begin{equation}
\Sigma_{\rm BS}(y;a) = \left| \sqrt{\frac{1}{a} \mathcal{J}_X(y;a,\varrho) -2y } -
\sqrt{\frac{1}{a} \mathcal{J}_X(y;a,\varrho) } \right| \,,\quad
y > y_R \mbox{ and } y<-\frac12,
\end{equation}
We choose the $y > y_R, y<-\frac12 $ branch since as $a\to 0$, the product
$\sqrt{a} y$ can be constant only if $y\to \infty$.
Expanding this result for $a\to 0$ we get 
\begin{equation*}
\Sigma_{\rm BS}(y;a)= 
\sqrt{\frac{1}{a} \mathcal{J}_X(y;a,\varrho) }
\left( 1 - \sqrt{1 - \frac{2y a }{\mathcal{J}_X(y;a,\varrho)}}  \right) 
= \frac{\sqrt{a} y}{\sqrt{ \mathcal{J}_X(y;a,\varrho)}} + O(a^{3/2} y^2) \nonumber \,.
\end{equation*}
Taking the $a\to 0$ limit at fixed $\sqrt{a} y = \sqrt2 \zeta$ gives
\begin{equation}
\lim_{a\to 0} \Sigma_{\rm BS}(y;a) =
\frac{\sqrt2 \zeta}{\sqrt{J(\zeta;\varrho)}}\,,
\end{equation}
where $J(\zeta;\varrho)$ is given by the limit (\ref{Jdef}).

\end{proof}

Substituting the expansion (\ref{Jzexp}) into (\ref{limSigBS}) reproduces the 
first three terms in the expansion of the celebrated analytical formula for
the implied volatility in the SABR model in the
short maturity asymptotic limit \cite{SABR}
\begin{equation}\label{Hagan}
\sigma_{\rm BS}(x,T) = \sigma_0 \frac{\zeta}{D(\zeta;\varrho)} 
\left(1 + \left( \frac14 \varrho \omega_0 \sigma_0 + \frac{1}{24} (2-3\varrho^2)
\omega^2 \right) T 
+ O(T^2)\right) \,.
\end{equation}
with
\begin{equation}
D(\zeta;\varrho) := 
\log \frac{\sqrt{1+2\zeta\varrho + \zeta^2} + \zeta + \varrho}{1+\varrho}\,,
\end{equation}
and the $O(T)$ terms holds only at the at-the-money (ATM) point $x=0$ 
\cite{Paulot}. Assuming $J(\zeta;\varrho) = \tilde J(\zeta;\varrho)$ reproduces 
the first factor in (\ref{Hagan}).

The result (\ref{Hagan}) is the
leading order term in a short maturity expansion, and the next two
terms in this expansion have been subsequently derived by 
Henry-Labordere \cite{HLbook} and Paulot \cite{PaulotSABR}. 
The ATM limit of this result is $\sigma_{\rm BS}(0,0) = \sigma_0$. 

In Figure~\ref{Fig:AL} we compare the asymptotic result (colored curves) 
with the SABR formula (\ref{Hagan}) (dashed black curve), for the model 
parameters $\sigma_0=0.2,\omega=1,\varrho=-0.75$, for several maturities
$T=0.25 - 5.0$. For sufficiently small maturity $T$, corresponding to 
small values of the $a=2(\omega^2 T)(\sigma_0^2 T)$ parameter,
the asymptotic result agrees very well with the short maturity limit (\ref{Hagan}).


\subsection{Short maturity expansion for the ATM implied volatility}
\label{sec:6.2}

We study here the $T\to 0$ expansion of the asymptotic implied volatility 
at the ATM point.

\begin{proposition}\label{prop:JX0}
(i) The first few terms in the expansion of the ATM asymptotic implied volatility 
$\sigma_{\rm BS}(0,T)$ in powers of $a =2\sigma_0^2\omega^2 T^2$ are
\begin{align}\label{SigATMT2}
\frac{1}{\sigma_0}\sigma_{\rm BS}(0,T) = \Sigma_{\rm BS}(0;a) 
&= 1 + \frac{\varrho}{4\sqrt2} \sqrt{a} +
\left( - \frac{1}{48} + \frac{1}{16} \varrho^2 \right) a + O(a^{3/2})\\
&= 1 +\frac14\varrho \omega\sigma_0 T
- \frac{1}{24} (1-3\varrho^2) \sigma_0^2\omega^2 T^2
+ O(T^3)\,.\nonumber
\end{align}

(ii) In the uncorrelated limit $\varrho=0$ the expansion contains only integer
powers of $a$
\begin{equation}
\Sigma_{\rm BS}(0;a) = 1 - \frac{1}{48} a + \frac{43}{23040} a^2 -
\frac{1907}{7741440} a^3 - \frac{51083}{7431782400} a^4 + O(a^5)\,.
\end{equation}
which means that the implied volatility contains only even powers of $T$
\begin{equation}\label{JX0aTaylor2}
\frac{\sigma_{\rm BS}(0,T)}{\sigma_0} = 1 - \frac{1}{24} \sigma_0^2\omega^2 T^2
+ \frac{43}{5760} \sigma_0^4\omega^4 T^4 + O(T^6)\,.
\end{equation}
\end{proposition}

\begin{proof}
(i) Taking $x=0$ in Theorem~\ref{thm:main} gives the ATM asymptotic implied
volatility $\Sigma_{\rm BS}(0;a) = 2\sqrt{\frac{1}{a} \mathcal{J}_X(0;a,\varrho)}$.
The rate function $\mathcal{J}_X(0;a,\varrho)$ is given by the extremal problem
(\ref{JXinf}). This can be expanded as $a\to 0$ by expanding the minimizers
in this problem around $u=v=1$. Denoting the minimizers $u_*,v_*$, this
expansion reads
\begin{equation}
\epsilon= \log u_*  = a_1 \sqrt{a} + a_2 a + O(a^{3/2})\,,\quad
\eta = \log v_*  = b_1 \sqrt{a} + b_2 a + O(a^{3/2})\,.
\end{equation}
Using the expansion (\ref{Iquart}) for the rate function $I(u,v)$,
the coefficients $a_i,b_i$ can be determined recursively. Substituting
into (\ref{JXinf}) gives the expansion of the rate function 
\begin{equation}
\mathcal{J}_X(0;a,\varrho)  =\frac14 a + \frac{\varrho}{8\sqrt2} a^{3/2}
+ \frac{1}{384} (-4+15\varrho^2) a^2 + O(a^{5/2})\,.
\end{equation}
Finally, substituting into 
$\Sigma_{\rm BS}(0;a) = 2\sqrt{\frac{1}{a} \mathcal{J}_X(0;a,\varrho)}$
gives the expansion (\ref{SigATMT2}).

(ii) The rate function $\mathcal{J}_X(y;a,0)$ is given in closed form 
by Eq.~(\ref{c:JX2})
\begin{equation}\label{JX0lam}
\mathcal{J}_X(0;a,0) = 2\lambda(\tan\lambda - \lambda) 
+ a \frac{\sin 2\lambda}{2\lambda} = 
2\lambda(2\tan\lambda  -\lambda)\,,
\end{equation}
where $\lambda$ is the solution of the equation
\begin{equation}\label{eq:a}
\frac{\lambda^2}{\cos^2\lambda} = \frac{a}{8}\,.
\end{equation}
The last equality in (\ref{JX0lam}) follows by substituting $a$ from 
(\ref{eq:a}).

Expanding the solution for $\lambda$ in powers of $\sqrt{a}$ and substituting
into (\ref{JX0lam}) gives
\begin{equation}
\mathcal{J}_X(0;a,0) =\frac14 a - \frac{1}{96} a^2 + \frac{1}{960} a^3 - 
\frac{23}{161280} a^4 + O(a^5)\,.
\end{equation}
This can be translated as before into an expansion for the asymptotic ATM
implied volatility.

\end{proof}

We can compare these expansions with the results in the literature. 
The $O(T)$ term in (\ref{SigATMT2}) coincides with the first $O(T)$ term in the
short maturity expansion of the implied volatility (\ref{Hagan}).
This expansion has been extended to $O(T^2)$ by Paulot \cite{Paulot},
where the $O(T^2)$ term was evaluated partially numerically. 
A closed form result for the ATM implied volatility expansion in the 
log-normal ($\gamma=1$) SABR model to $O(T^2)$ has been communicated to us 
by Alan Lewis \cite{LewisUnpublished}
\begin{eqnarray}\label{ALfull}
&& \frac{1}{\sigma_0} \sigma_{\rm BS}(0,T) =
1 + \frac{1}{24} \sigma_0 \omega T
\left[6 \varrho  + \frac{\omega}{\sigma_0} (2- 3 \varrho^2) \right]  \\
&& + \frac{1}{1920} \omega^2 \sigma_0^2 T^2
\left[
(-80 + 240 \varrho^2 ) +
\frac{\omega}{\sigma_0} \varrho (240 - 180 \varrho^2) +
\frac{\omega^2}{\sigma_0^2} (-12 + 60 \varrho^2 - 45 \varrho^4) 
\right] +  O(T^3) \,.\nonumber
\end{eqnarray}

The short maturity expansion of the implied volatility in 
a wide class of stochastic volatility models called MAP-like 
(Markov Additive Processes)
which include $\gamma=1$ SABR, is known (\cite{Lewis2017} page 505) 
to admit a double series expansion in $(x=\log(K/S_0),T)$.

Recalling that the asymptotic limit considered
in our paper corresponds to $\omega^2 T  \ll 1, \sigma_0 \omega=O(1)$,
it is easy to see that (\ref{SigATMT2}) is reproduced by taking this
limit in (\ref{ALfull}).

The existence of the limit $\omega^2 T \to 0, \sigma_0 \omega_0 = O(1)$
constrains the form of the higher order terms in the ATM implied volatility
\begin{equation}
\frac{\sigma_{BS}(0,T)}{\sigma_0} = 
1 + \left(c_1^{(1)} \omega\sigma_0 + c_1^{(0)} \omega^2\right) T + 
\sum_{k=2}^\infty \sigma_k(\omega,\sigma_0,\varrho) T^k \,.
\end{equation}
Assume that $\sigma_k(\omega,\sigma_0,\varrho) T^k$ is a polynomial in 
$(\omega \sqrt{T}), (\sigma_0 \sqrt{T})$ of order $2k$, this must have the form
\begin{equation}
\sigma_k(\omega,\sigma_0,\varrho) = \sum_{j=0}^{k} c_k^{(j)}(\varrho) 
\omega^{2k-j} \sigma_0^{j}\,.
\end{equation}
For example, a $O(T^2)$ term of the form $\omega \sigma_0^3 T^2$ 
is not allowed, as it diverges in the limit considered.

In the limit considered,
only the terms proportional to $c_k^{(k)}$ contribute. The expansion of
Proposition~\ref{prop:JX0} determines these coefficients.
For example we get $c_{4}^{(4)}(\varrho)= \frac{1}{5760}(43-375\varrho^2+315\varrho^4)$ 
for the coefficient of $(\omega \sigma_0 T)^4$, which has been confirmed 
by explicit computation \cite{LewisUnpublished}.


\subsection{Extreme strikes asymptotics}
\label{sec:6.4}

We study here asymptotics in the extreme strikes region $|x| \to \infty$
for the uncorrelated case $\varrho=0$.
Since the implied volatility is symmetric in $x$ for this case, 
see Remark \ref{remark:symm}, 
it will be sufficient to study asymptotics for large strike $x\to \infty$.

\begin{proposition}
In the large log-strike region $x\to \infty$, the asymptotic
volatility (\ref{SigBSdef}) in  the uncorrelated log-normal SABR model
$\varrho=0$ has the expansion
\begin{align}\label{SigExp}
\Sigma_{\rm BS}(x;a,0) 
&= \sqrt{2x} - \frac{1}{\sqrt{2a}} \log (2x) - \frac{1}{2a}
\log(2 \log(2x)) + \frac{1}{\sqrt{2a}}
\\
& \quad + \frac{1}{4a\sqrt{2x}} \log^2 (2x) + \frac{1}{2a\sqrt{2x}} 
\log(2x) \log(2\log(2x)) + 
O\left(\frac{\log(2\log(2x))}{\log(2x)}\right)
\nonumber \,.
\end{align}
\end{proposition}

\begin{proof}
We use the result of Proposition~\ref{prop:JXasympt} for the large argument
limit of the rate function to obtain the asymptotics of $\Sigma_{\rm BS}(x;a,0)$
for $x \to \infty$
\begin{align}
\Sigma_{\rm BS}(x;a,0) 
&= \Sigma_{\rm BS}(-x;a,0) = 
\sqrt{\frac{1}{a} J_X(x;a,0)} - \sqrt{-2x + \frac{1}{a} J_X(x;a,0)}
 \\
&= \sqrt{2x}\left( \sqrt{1 + r(x,a)} - \sqrt{r(x,a)} \right)
\nonumber
\\
&= \sqrt{2x} \left( 1 - \sqrt{r(x,a)} + \frac12 r(x,a)
+ O(r^2(x,a)) \right)\,, \nonumber
\end{align}
where
\begin{equation}
r(x,a) = \frac{1}{4ax} \log^2(2x) + \frac{1}{2ax} \log(2x) \log(2 \log(2x))
- \frac{1}{2ax} \log(2x) + O\left(\frac{\log\log x}{2ax}\right)\,.
\end{equation}
Expanding the result gives (\ref{SigExp}).
\end{proof}

\begin{corollary}
The large strike asymptotics of the implied variance for $x \to \infty$
at fixed $T$ is given by
\begin{equation}
\sigma_{\rm BS}^2(x,T) T= 
\sigma_0^2 \Sigma_{\rm BS}^2
\left( \frac{x}{\sigma_0^2 T},a,0\right) T =
\frac{2x}{\left(\sqrt{1+r\left(\frac{x}{\sigma_0^2 T},a\right)} 
 + \sqrt{r\left(\frac{x}{\sigma_0^2 T},a\right)}\right)^2} \,,
\end{equation}
with $a = 2(\omega^2 T)(\sigma_0^2 T)$ and 
\begin{equation}
r\left(\frac{x}{\sigma_0^2 T},a\right) = \frac{1}{4\omega^2 T x}
\left\{ \frac12 \log^2 \left( \frac{2x}{\sigma_0^2 T} \right) + 
\log\left( \frac{2x}{\sigma_0^2 T} \right)
\log\left( 2\log \left(\frac{2x}{\sigma_0^2 T} \right)\right) -
\log\left( \frac{2x}{\sigma_0^2 T} \right) \right\}\,.
\end{equation}
Using the expansion $(\sqrt{1+r} + \sqrt{r})^{-2} = 1 - 2\sqrt{r} + 2r + O(r^{3/2})$
we get, keeping only the $O(\sqrt{r})$ term,
\begin{align}\label{impvol2}
\sigma_{\rm BS}^2(x,T)T &= 2x \left( 1 - \frac{1}{\sqrt{2\omega^2 T x}}
\sqrt{L^2 + 2 L \log(2L) - 2L} + \cdots \right) \\
&= 2x - \frac{\sqrt{2x}}{\sqrt{\omega^2 T}} 
\sqrt{L^2 + 2 L \log(2L) - 2L} + \cdots\,, \nonumber
\end{align}
where we denoted $L=\log\left( \frac{2x}{\sigma_0^2 T}\right)$.
\end{corollary}

The leading term in (\ref{impvol2}) agrees with the result expected from Lee's moment 
formula. Recall that under the Log-Euler-log-Euler scheme, all moments $\mathbb{E}[(S_n)^{1+\varepsilon}]$ with
$\varepsilon>0$ are infinite \cite{PZStochVol}. The Lee moment formula 
\cite{Lee} predicts that the large strike asymptotics of the 
implied variance is $\lim_{x\to \infty} \sigma_{\rm BS}^2(x,T) T = 2x$. 

The short maturity SABR formula (\ref{Hagan}) gives an implied volatility which 
grows faster than
the behavior expected from the Lee's moment formula. Therefore its applicability
is limited to a region of log-strikes sufficiently close to the at-the-money
region.

The result (\ref{impvol2}) agrees with the subleading correction
derived by Gulisashvili and Stein in the uncorrelated Hull-White model 
\cite{GSHW}. In Theorem 3.1 and Corollary 3.1 of this paper, the following
asymptotic result is proved in this model (assuming $S_0=1$)
\begin{equation}
\sigma_{\rm BS}^2(K,T) T = 2\log K - \frac{1}{\omega\sqrt{T}} \sqrt{2\log K}
( \log\log K + \log\log\log K) + \cdots\,.
\end{equation}
The leading correction term to the Lee's moment formula $\sim \log\log K$ 
agrees with (\ref{impvol2}).

\section{Numerical benchmarks}
\label{sec:7}

We compare in this section numerical benchmarks for implied volatility in the
$\gamma=1$ SABR model, with the asymptotic results of this paper.

\begin{table}
\caption{\label{Tab:corr} 
Scenarios for the numerical testing, from Table 8.6 in \cite{Lewis2017}.
The model parameters are $\sigma_0=0.2,\omega=1.0, \varrho=-0.75$.
The table shows the parameter $a=2(\omega^2 T)(\sigma_0^2T)$ 
of the asymptotic expansion, the point $(u_m,v_m)$ determined by Corollary
\ref{lemma:2} and the right switch point $y_R$ given by (\ref{yRdef}).}
\begin{center}
\begin{tabular}{|c|c|c|c|c|}
\hline
$T$ & $a$ & $\sigma_0^2 T$ & $(u_m,v_m)$ & $y_R$ \\
\hline \hline
0.25 & 0.005 & 0.01 & (0.9636, 0.9638) & 0.4821 \\
1.0  & 0.08  & 0.04 & (0.8664, 0.8692) & 0.4362 \\
2.0  & 0.32  & 0.08 & (0.7589, 0.7676) & 0.3882 \\
5.0  & 2.0   & 0.20 & (0.5360, 0.5636) & 0.2938 \\
\hline 
\end{tabular}
\end{center}
\end{table}

The benchmark option prices are taken from Table~8.6 in \cite{Lewis2017}.
They were obtained using the transform method of
\cite{Lewis,Lewis2017} with the model parameters $\sigma_0=0.2,\omega=1.0,
\varrho=-0.75$ and several option maturities $T=0.25,1,2,5$. 

The asymptotic result $\sigma_{\rm BS}(x,T)/\sigma_0$ of Theorem~\ref{thm:main}
(blue/red curve) is compared against the benchmark values in 
Figure~\ref{Fig:AL} (black dots). The three regions of Theorem~\ref{thm:main}
are shown in different colors (red for the central region $-\frac12\sigma_0^2 T
\leq x \leq y_R \sigma_0^2 T$).
The figures show also the leading order $O(T^0)$ short-maturity 
asymptotics in the SABR model (\ref{Hagan}) as the dashed black curve.

From these results we note the following observations:

(i) For short maturities the agreement of the asymptotic result
with the SABR asymptotic formula (\ref{Hagan}), and with the numerical benchmark
results is very good. The central region of log-strikes of 
Theorem~\ref{thm:main} is very small, and it expands as the maturity increases.

(ii) At larger maturities the short-maturity approximation
(\ref{Hagan}) overestimates the actual implied volatility. While the asymptotic
result reproduces the decreasing trend of the numerical result, 
it is an overestimate for longer maturities. 

As explained in the previous section, the 
asymptotic result holds in the limit $\sigma_0/\omega \gg 1$. The numerical
benchmarks considered have $\sigma_0/\omega = 0.2$ which is not particularly
large. The agreement is expected to become better if this
ratio is large, corresponding to a small vol-of-vol scenario.
This is confirmed by the results in Table~\ref{tab:3} where the 
asymptotic result for the ATM implied volatility is compared with numerical 
benchmarks for a scenario with $\sigma_0=1.0,
\omega=0.1$. The agreement improves in the latter case, as expected.

(iii) From Table~\ref{tab:3} one observes that in the uncorrelated
case $\varrho=0$, the actual ATM implied volatility 
has a non-monotonic dependence on maturity: starts at $\sigma_0$ as $T\to 0$, 
first increases with maturity, 
and then decreases as $T\to \infty$. On the other hand, the asymptotic result 
has a monotonously decreasing trend. 

The discrepancy between the two results at
short maturity can be traced back to the absence of a $O(T)$ term
in the asymptotic expansion for the uncorrelated case, which is responsible 
for the increasing trend of the numerical results for small maturity. Using the
full $O(T^2)$ expansion for the ATM implied volatility (\ref{ALfull}),
which includes this term, 
reproduces well the benchmark results, as shown in Table~\ref{tab:3}.

\begin{figure}
    \centering
   \includegraphics[width=2.5in]{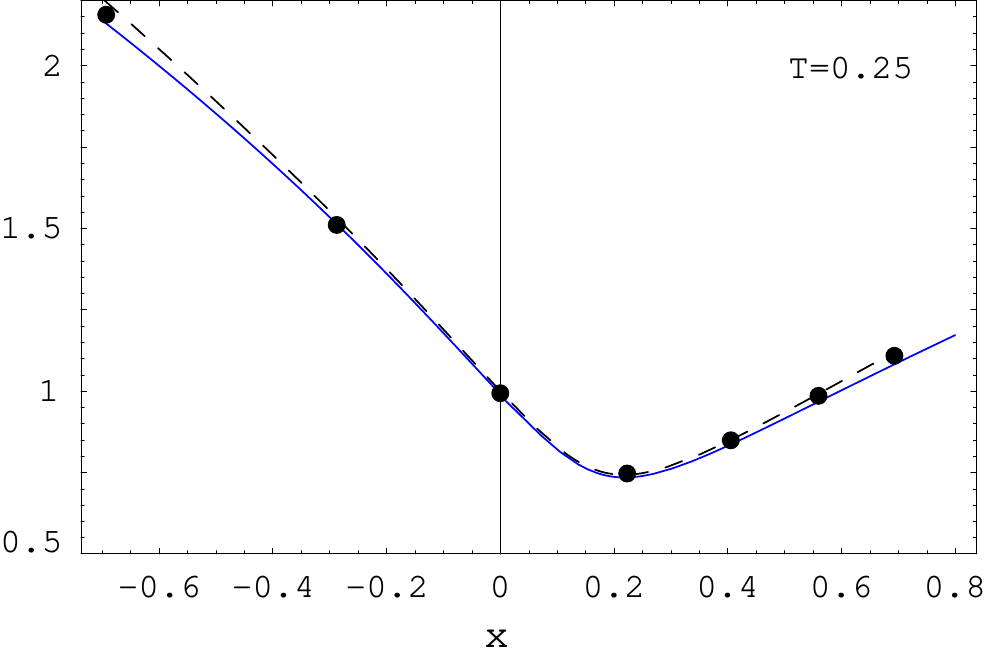}
   \includegraphics[width=2.5in]{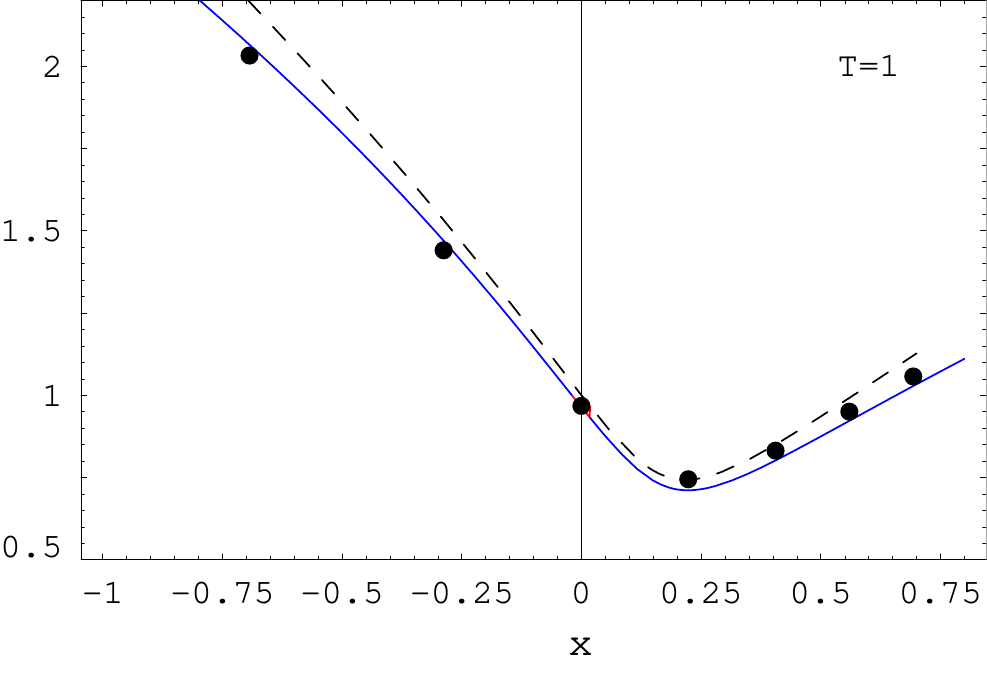}
   \includegraphics[width=2.5in]{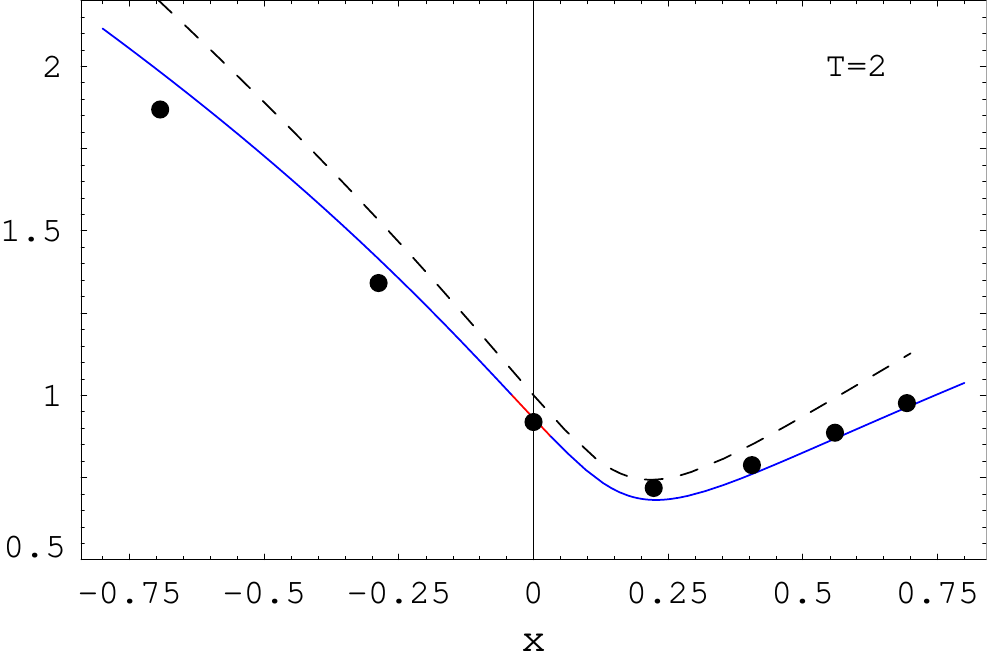}
   \includegraphics[width=2.5in]{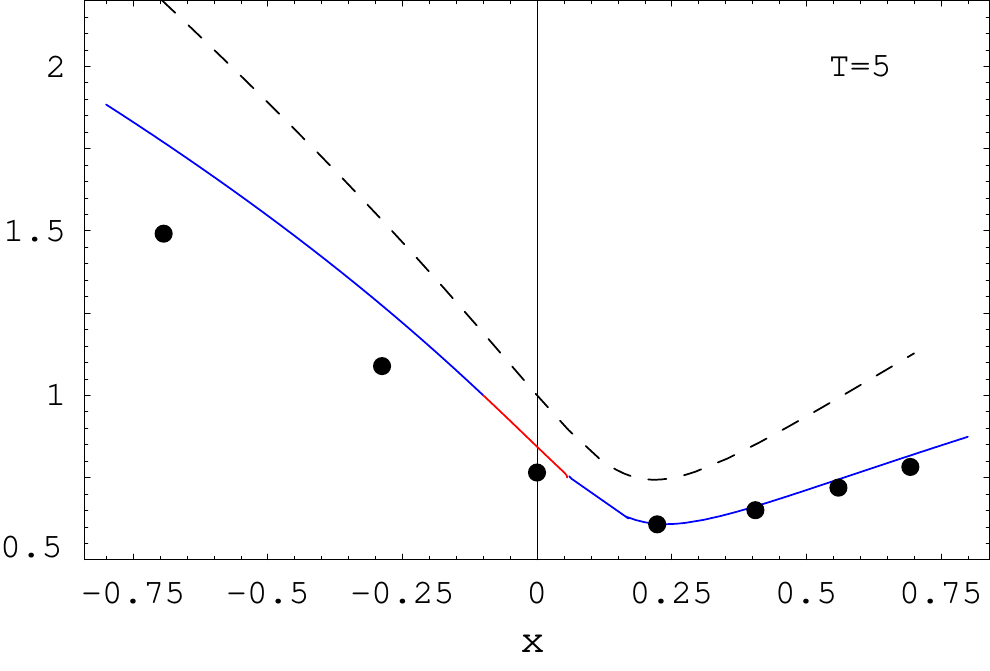}
    \caption{Plots of the asymptotic implied volatilities 
$\sigma_{\rm BS}(x,T)/\sigma_0$ vs $x=\log(K/S_0)$ (colored curves) 
for the scenarios in Table~\ref{Tab:corr}, taken from Table~8.6 of
\cite{Lewis2017}. The black dots show the benchmark values from \cite{Lewis2017},
and the dashed black curve shows the short maturity SABR implied volatility.
The different regions in Theorem~\ref{thm:main} are shown in different colors
(blue/red).}
\label{Fig:AL}
 \end{figure}

\begin{table}
\caption{\label{tab:3} 
Benchmark numerical values computed using the transform method (AL) 
\cite{Lewis2017}
and the second order short maturity expansion $O(T^2)$ from (\ref{ALfull}) 
for maturities $T=0.25, 1, 2, 5, 50$, comparing with the asymptotic results
from Theorem~\ref{thm:main}.
The model parameters are $(\sigma_0,\omega) = (0.2,1.0), (1.0,0.1)$ and 
$\varrho=0$. }
\begin{center}
\begin{tabular}{|c|ccc|ccc|}
\hline
  & \multicolumn{3}{|c|}{$\sigma_0=0.2,\omega=1.0$} 
  & \multicolumn{3}{|c|}{$\sigma_0=1.0,\omega=0.1$} \\
\hline
$T$ & AL & $O(T^2)$ expansion & asymptotic 
    & AL & $O(T^2)$ expansion & asymptotic \\
\hline
\hline
0.25 & 0.20407  & 0.204068 & 0.19998
     & 1.00018  & 1.00018  & 0.99997 \\
1.0  & 0.21460  & 0.215083 & 0.19967
     & 1.00041  &  1.00042 & 0.99958 \\
2.0  & 0.22123 & 0.227 & 0.19870
     & 0.999974 & 0.999998 & 0.99835 \\
5.0  & 0.20451 & 0.24375 & 0.19286
     & 0.993662 & 0.993734 & 0.99002 \\
50.0 & 0.07822 & -2.925 & 0.11275
     & 0.719669 & -0.0015625 & 0.72071 \\
\hline 
\end{tabular}
\end{center}
\end{table}

\appendix

\section{The zero correlation case}
\label{sec:app}

The rate function $\mathcal{I}_X(x;\varrho)$ simplifies in the 
uncorrelated limit $\varrho=0$, and can be expressed
in closed form. This result can be used to derive the asymptotics of the rate 
function in various limits of small/large arguments.
We give in this Appendix these results and their proofs.

\subsection{The rate function $J_{\rm BS}(x)$}

The rate function appearing in the LDP for $\mathbb{P}(\frac{1}{n}V_n\in\cdot)$ 
can be extracted from Proposition 6 in \cite{PZDiscreteAsian}. 
A simpler form is given in Corollary 13 of the same
paper, in terms of the function $J_{\rm BS}(x)$. 
We reproduce here this result for ease of reference.

\begin{proposition}\label{prop:JBS}
Define $A_n := \sum_{i=0}^{n-1} e^{sZ_i + (m-\frac12 s^2) t_i}$, with $Z_i$ a 
standard
Brownian motion sampled on uniformly distributed times $t_i = \tau i$. 
Consider the $n\to \infty$
limit at fixed $b=\frac12 s^2 \tau n^2$ and $r = m\tau n$. In this limit 
$\mathbb{P}(\frac{1}{n} A_n \in \cdot) $ satisfies a LDP with rate function 
$I_{\rm BS} (\cdot ) = \frac{1}{2b} J_{\rm BS}(\cdot)$.
For $r=0$, 
the rate function $\mathcal{J}_{\rm BS}(x)$ is given by
\footnote{The case $r=0$ covers the case considered in this paper. See Proposition 6
in \cite{PZDiscreteAsian} for the general $r\neq 0$ case.}
\begin{equation}\label{JBSdef}
\mathcal{J}_{\rm BS}(x) = 
\begin{cases}
\frac12\xi^2 - \xi \tanh(\xi/2) &
\text{for $x \geq 1$}, \\
2\lambda(\tan \lambda - \lambda) &
\text{for $0 <x \leq 1$}, \\
0 & \text{for $x=1$},
\end{cases}
\end{equation}
where $\xi>0$ is the unique solution of the equation
\begin{equation}\label{eq:xi}
\frac{1}{\xi} \sinh\xi = x\,,
\end{equation}
and $\lambda\in (0,\frac{\pi}{2})$ is the unique solution of the equation
\begin{equation}\label{eq:lambda}
\frac{1}{2\lambda} \sin(2\lambda) = x\,.
\end{equation}
\end{proposition}

The sum $V_n= \sum_{i=0}^{n-1} \sigma_i^2 \tau$ is obtained by identifying 
$\frac{V_n}{\sigma_0^2\tau}
\to A_n$ with the substitutions $s \to 2\omega, m \to \omega^2$. In the $n\to \infty$ limit at fixed
$b = \frac12s^2 \tau n^2 = 2 \omega^2 \tau n^2$, it is clear that $r = m n\tau = \omega n \tau \to 0$. 
This justifies the $r=0$ limit used in (\ref{JBSdef}).

We will require also the derivative of the rate function 
$\mathcal{J}_{\rm BS}(x)$.
This can be obtained in closed form and is given by the following result.

\begin{corollary}
The derivative of the rate function $\mathcal{J}_{\rm BS}(x)$ is given by
\begin{equation}
\mathcal{J}'_{\rm BS}(x) = 
\begin{cases}
\frac12 \frac{\xi^2}{\cosh^2(\xi/2)} & \text{for $x > 1$}, \\
- 2 \frac{\lambda^2}{\cos^2\lambda} & \text{for $0 < x < 1$},
\end{cases}
\end{equation}
where $\xi$ is the solution of the equation (\ref{eq:xi}) and $\lambda$ is the solution of the
equation (\ref{eq:lambda}).
\end{corollary}

\subsection{Closed form result for the rate function $\mathcal{I}_X(x;\varrho=0)$}
\label{sec:a1}

The rate function  $\mathcal{I}_X(x;\varrho)$ giving the large deviations for 
$\mathbb{P}(\frac{1}{n}
\log S_n\in \cdot)$ is given by the solution of the extremal problem
in Proposition~\ref{prop:LDlogS} in the main text. 
In the zero correlation limit $\varrho=0$ this extremal problem simplifies 
to a one-dimensional extremal problem. Using the one dimensional projection 
relation $\inf_v I(u,v) = 2\mathcal{J}_{\rm BS}(u)$, the extremal problem 
(\ref{JXinf}) simplifies as
\begin{align}\label{1dimJX}
\mathcal{J}_X(y; a,0) &=\inf_{u,v} \left\{ \frac12 I(u,v) + \frac{a}{u}
\left( y + \frac12 u \right)^2 \right\} \\
&= \inf_{u} \left\{ \mathcal{J}_{\rm BS}(u) + \frac{a}{u}
\left( y + \frac12 u \right)^2 \right\}
= \mathcal{J}_{\rm BS}(u_*) + \frac{a}{u_*}
\left( y + \frac12 u_* \right)^2\,, \nonumber
\end{align}
where we denoted in the last line the 
optimal value of $u$ in the extremal problem as $u_*(y)$. 

\begin{lemma}\label{lemma:1}
The extremal value $u_*(y)$ has the following properties:

(i) $u_*(y) > 1$ for $|y| > \frac12$;

(ii) $0 < u_*(y) < 1$ for $|y| < \frac12$.
\end{lemma}

\begin{proof}
The optimal value $u_*(y)$ is given by the solution of the equation
\begin{equation}\label{eq:ust}
\mathcal{J}_{\rm BS}'(u) + \frac{a}{u^2} \left( \frac14 u^2 - y^2\right) = 0\,.
\end{equation}
This equation can be written equivalently as
\begin{equation}
\mathcal{J}_{\rm BS}'(u) = a \left( \frac{y^2}{u^2} - \frac14 \right).
\end{equation}
The function on the right side is decreasing in $u$ and is positive at $y=1$ for $|y| > \frac12$,
and negative for $|y| < \frac12$. The function on the left side is increasing and vanishes at $y=1$.
This implies that the two sides will become equal at a point $u_*(y)$ which is larger than $1$ in the 
first case, and lower than $1$ in the second case. This proves the claim.
\end{proof}

We give next a closed form result for the rate function, which is useful
for numerical evaluations and deriving asymptotic expansions.

\begin{corollary}\label{corr:JXrho0}
In the zero correlation limit $\varrho = 0$, 
the rate function $\mathcal{J}_X(y;a,0)$ has the following explicit form.

Case 1. $|y|>\frac12 $.
\begin{equation}\label{c:JX1}
\mathcal{J}_X(y;a,0) = \frac12 \xi^2 - \xi\tanh(\xi/2) + a 
\frac{\xi}{\sinh \xi} \left(y + \frac{1}{2\xi}\sinh\xi\right)^2\,,
\end{equation}
where $\xi\in(0,\infty)$ satisfies the equation
\begin{equation}\label{c:eq:xi}
\frac{1}{4a} \frac{\xi^2}{\cosh^2(\xi/2)} + \frac18 - 
\frac12 y^2 \frac{\xi^2}{\sinh^2\xi} = 0\,.
\end{equation}

Case 2.  $|y|<\frac12 $. 
\begin{equation}\label{c:JX2}
\mathcal{J}_X(y;a,0) = 2\lambda (\tan\lambda - \lambda) + a 
\frac{2\lambda}{\sin 2\lambda} 
\left(y + \frac{1}{2} \frac{\sin 2\lambda}{2\lambda}\right)^2\,,
\end{equation}
where $\lambda\in(0,\pi/2)$ satisfies the equation
\begin{equation}\label{c:eq:lam}
-\frac{1}{a} \frac{\lambda^2}{\cos^2\lambda} + \frac18 - 
\frac12 y^2 \frac{(2\lambda)^2}{\sin^2(2\lambda)} = 0\,.
\end{equation}

Case 3. $|y|=\frac{1}{2}$.
\begin{equation}
\mathcal{J}_{X}\left(-\frac{1}{2};a,0\right)=0,
\qquad
\mathcal{J}_{X}\left(\frac{1}{2};a,0\right)=a.
\end{equation}
\end{corollary}

\begin{proof}
Case 1. $|y| > \frac12$. For this case $u_*>1$ so we use the corresponding branch of $\mathcal{J}_{\rm BS}(u)$, with $u_* = \frac{\sinh\xi}{\xi}$. The value of $\xi$ is determined by expressing (\ref{eq:ust})
as an equation for $\xi$. This reproduces (\ref{c:eq:xi}).
The rate function is obtained by substituting $u_*$ into (\ref{1dimJX}). This reproduces (\ref{c:JX1}).

Case 2. $|y| < \frac12$. For this case $0 < u_* < 1$ and we use the second branch of 
$\mathcal{J}_{\rm BS}(u)$, with $u_* = \frac{\sin(2\lambda)}{(2\lambda)}$. The equation for $\lambda$
is obtained by substituting the expression for $u_*$ into (\ref{eq:ust}). This reproduces (\ref{c:eq:lam}).
Proceeding in a similar way we get the rate function by substituting $u_*$ into (\ref{1dimJX}),
which reproduces (\ref{c:JX2}).
\end{proof}

\begin{proposition}\label{prop:Isymm}
The rate function in the uncorrelated case satisfies the symmetry relation
\begin{equation}
J_X(y;a,0) - J_X(-y;a,0) = 2ay\,.
\end{equation}
\end{proposition}

\begin{proof}
This follows by noting that the optimizer $u_*(y)$, the solution of the equation (\ref{eq:ust}) is a symmetric function in $y$.  If $u_*(y)$ is a solution of this equation
for a certain $y$, it will be also a solution for $-y$. Thus we have
\begin{equation}
J_X(y;a,0) - J_X(-y;a,0) = \frac{a}{u_*} \left(y+\frac12 u_*\right)^2 - \frac{a}{u_*}\left(-y + \frac12 u_*\right)^2 = 2ay\,.
\end{equation}
\end{proof}


\subsection{Asymptotic expansion for $\mathcal{J}_X(x;a,0)$}

We derive here the asymptotic expansion for the rate function
$\mathcal{J}_X(x;a,0)$ for very large argument $|x| \to \infty$.

\begin{proposition}\label{prop:JXasympt}
As $x\to \infty$ the rate function $\mathcal{J}_X(x;a,0)$
has the expansion
\begin{equation}\label{Jlargex}
\mathcal{J}_X(x;a,0) = 2a x + \frac12 \log^2(2x) 
+ \log(2x)\log(2\log(2x)) 
- \log(2x) + O(\log\log x).
\end{equation}
\end{proposition}

\begin{proof}
We use the explicit form for the rate function $\mathcal{J}_X(x;a,0)$ given in
Corollary~\ref{corr:JXrho0}. This has two branches, for $|x|>\frac12$ and
$|x|<\frac12$. We are interested in $x\gg 1$ so we give below only
the result for $x>\frac12$ for ease of reference.
\begin{equation}\label{JXsol1}
\mathcal{J}_X(x;a,0) = \frac12\xi^2 - \xi \tanh(\xi/2) + a \frac{\xi}{\sinh\xi}
\left( x + \frac{1}{2\xi} \sinh\xi \right)^2\,,
\end{equation}
where $\xi$ is the solution of the equation
\begin{equation}\label{eq:xi1}
\frac{1}{4a} \frac{\xi^2}{\cosh^2(\xi/2)} + \frac18 - \frac12 x^2 \frac{\xi^2}
{\sinh^2 \xi} = 0\,.
\end{equation}

The first two terms in (\ref{JXsol1}) correspond to $\mathcal{J}_{\rm BS}(y)$ in
(\ref{JXdef}). The asymptotics of $\mathcal{J}_{\rm BS}(y)$ for $y\to \infty$
was obtained in Proposition 13 in \cite{PZSIFIN}. We will follow a similar
approach to obtain the asymptotics of $\mathcal{J}_{X}(x;a,0)$ for $x\to \infty$.

The strategy will be to invert the equation (\ref{eq:xi1}) for $x \gg 1$ and
use the resulting expansion for $\xi(x)$ into (\ref{JXsol1}).
First we write the equation (\ref{eq:xi1}) for $\xi$ as
\begin{eqnarray}\label{eq:xi2}
&& \frac{1}{2a} \frac{\sinh^2\xi}{\cosh^2(\xi/2)} + \frac14 \frac{\sinh^2\xi}{\xi^2}
= x^2\,, \\
\label{eq:xi3}
&& \frac{2}{a} \sinh^2(\xi/2) + \frac14 \frac{\sinh^2\xi}{\xi^2}
= x^2\,, \\
\label{eq:xi4}
&& \sinh^2(\xi/2)\left( \frac{2}{a} + \frac{\cosh^2\xi}{\xi^2} \right)
= x^2 \,.
\end{eqnarray}

We solve this equation for $\xi$ as $x\to \infty$ using asymptotic inversion,
see for example Sec.1.5 in \cite{Olver}. 
We first write the equation (\ref{eq:xi4}) for $\xi$ as 
\begin{equation}\label{eq:xi5}
e^\xi \left(1 - e^{-\xi}\right)^2 \left( \frac{1}{2a} + \frac{1}{16\xi^2}
e^\xi \left(1 + e^{-\xi}\right)^2 \right) = x^2\,.
\end{equation}
Take logs of both sides
\begin{equation*}
\xi + 2 \log(1-e^{-\xi}) + \xi + 2 \log(1+e^{-\xi}) - 2 \log(4\xi)
+ \log\left( 1 + \frac{8\xi^2 e^{-\xi}}{a(1+e^{-\xi})^2} \right) = \log x^2\,,
\nonumber
\end{equation*}
or equivalently
\begin{equation*}
2\xi =  \log x^2 - 2 \log\left(1-e^{-2\xi}\right) + 2 \log(4\xi) -
\log\left( 1 + \frac{8\xi^2 e^{-\xi}}{a(1+e^{-\xi})^2}\right)\,.
\end{equation*}

As $x\to \infty$, this is approximated as $2\xi=\log x^2 + O(\log \xi)$.  
This estimate can be improved by iteration, starting with this first order
approximation and solving for $\xi^{(i+1)}$ by inserting the previous 
iteration on the right-hand side. The first two iterations are
\begin{eqnarray}
&& 2\xi^{(1)} = \log x^2 + O(\log\log x^2)\,, \\
&& 2\xi^{(2)} = \log x^2 + 2 \log(2 \log x^2) + 
O\left(\frac{\log\log x^2}{\log x^2}\right)\,.
\end{eqnarray}
To this order, the dependence on $a$ is of higher order. This means that 
we can approximate the equation (\ref{eq:xi4}) for $\xi$ 
with $\frac{\sinh \xi}{\xi}=2x$ (by neglecting the $a$ term)
and we can read off the solution from the Prop.~13 in \cite{PZSIFIN}
by replacing $K/S_0 \to 2x$.

Thus we get the 
expansion of the first 2 terms in the rate function (\ref{JXsol1})
(given by $\mathcal{J}_{\rm BS}(y_*)$) from Prop.~13 in \cite{PZSIFIN}
\begin{align}
\mathcal{J}_{\rm BS}(y_*) &=
\frac12 \log^2 (2x) + \log(2x) \log(2\log(2x)) - \log(2x)\\
& 
\qquad + 
3 \log^2(2\log(2x)) - 2 \log(2\log(2x)) + O(\log^{-1}(2x))\,.
\nonumber
\end{align}
The second term in (\ref{JXsol1}) is
\begin{equation}
a \frac{\xi}{\sinh \xi} \left(x + \frac{1}{2\xi} \sinh\xi\right)^2 = 
\frac{a}{2x} \left(x +  \frac12 2x\right)^2 = 2ax\,.
\end{equation}
Adding them gives
\begin{equation*}
\mathcal{J}_X(x;a,0) = 2ax + \frac12 \log^2 (2x) + 
\log(2x) \log(2\log(2x))
- \log(2x) + O(\log\log x) \,.\nonumber
\end{equation*}
This completes the proof of Eq.~(\ref{Jlargex}).
\end{proof}

\section*{Acknowledgements}
We are grateful to Alan Lewis for communicating unpublished results, 
and for kindly providing benchmark numerical evaluations
in the SABR model which were used for comparison with the asymptotic 
results. 
Lingjiong Zhu is grateful to the partial support from NSF Grant DMS-1613164.

\end{document}